\documentclass[a4paper,UKenglish,cleveref, autoref, nolineno, thm-restate]{socg-lipics-v2021}

\hideLIPIcs  %

\graphicspath{{./graphics/}}%

\usepackage[disable]{todonotes}

\usepackage{color,latexsym,amsmath,amssymb}
\usepackage{fancyhdr}
\usepackage{amsthm}
\usepackage{svg}
\usepackage{changepage}
\usepackage{graphicx}
\usepackage[utf8]{inputenc}
\usepackage{listings}
\usepackage{xcolor}
\usepackage{tikz}
\usepackage[thinlines]{easytable}
\usepackage{makecell}
\usepackage{fontawesome5} %
\usepackage{dsfont}
\usepackage{subcaption}
\usepackage[export]{adjustbox}
\usepackage{algorithm}
\usepackage{algpseudocode}
\usepackage{arydshln}%
\newtheorem{hypothesis}{Hypothesis}[section]

\title{Subquadratic Approximation Algorithms for Separating Two Points with Objects in the Plane} %

\author{Jayson Lynch}{MIT}{jaysonl@mit.edu}{https://orcid.org/0000-0003-0801-1671}{} %

\author{Jack Spalding-Jamieson}{Independent}{jacksj@uwaterloo.ca}{https://orcid.org/0000-0002-1209-4345}{}

\authorrunning{J. Lynch and J. Spalding-Jamieson} %

\Copyright{Jayson Lynch and Jack Spalding-Jamieson} %



\category{} %

\makeatletter
\ifx\@hideLIPIcs\@undefined
    \relatedversion{https://arxiv.org/abs/2507.22293} %
\else\fi
\makeatother

\EventEditors{John Q. Open and Joan R. Access}
\EventNoEds{2}
\EventLongTitle{42nd Conference on Very Important Topics (CVIT 2016)}
\EventShortTitle{CVIT 2016}
\EventAcronym{CVIT}
\EventYear{2016}
\EventDate{December 24--27, 2016}
\EventLocation{Little Whinging, United Kingdom}
\EventLogo{}
\SeriesVolume{42}
\ArticleNo{23}

\newcommand{\RR}{\mathbb{R}}

\newcommand{\ZZ}{\mathbb{Z}}

\newcommand{\defn}[1]{\textbf{#1}}

\DeclareMathOperator{\OO}{\mathcal{O}}
\DeclareMathOperator{\lO}{\widetilde{\mathcal{O}}}

\newcommand{\cC}{\mathcal{C}}

\newcommand{\SP}{\texttt{Time}_{\text{SSSP}}}
\newcommand{\jaysonl}[1]{}
\newcommand{\jack}[1]{}

\newcommand{\remove}[1]{}

\begin{document}

\maketitle

\begin{abstract}
The
(unweighted)
\defn{point-separation} problem asks,
given a pair of points $s$ and $t$ in the plane,
and a set of candidate geometric objects,
for the
minimum-size
subset of objects whose union blocks all paths from $s$ to $t$.
Recent work has shown that the point-separation problem can be characterized as a type of shortest-path problem in a geometric intersection graph within a special lifted space.
However, all known solutions to this problem essentially reduce to some form of APSP,
and hence take at least quadratic time, even for special object types.

Recent work has also given a conditional lower bound showing
that, for many classes of objects, there is no $O(n^{\frac32-\varepsilon})$-time
algorithm for this problem.
We improve this lower bound under a different fine-grained hypothesis,
to show that for the same classes of objects,
there is no $O(n^{2-\varepsilon})$-time
algorithm for this problem, rendering many of the prior algorithms tight
up to sub-polynomial factors.
We also show that this same lower bound applies to
multiplicative $(1+\varepsilon)$-approximations.

Our main results are positive:
We bypass this barrier by instead considering additive $+1$ approximations,
and multiplicative-additive $(1+\varepsilon,+1)$-approximations.
That is, we produce solutions of size $\text{OPT}+1$ and $(1+\varepsilon)\text{OPT}+1$, respectively.
In this paradigm,
we are able to devise algorithms that are fundamentally different from the APSP-based approach.
In particular,
we give Monte Carlo randomized additive $+1$ approximation algorithms running in $\lO(n^{\frac32})$ time for disks,
axis-aligned line segments and constant-complexity rectilinear polylines,
and
$\lO(n^{\frac{11}6})$ time for line segments and constant-complexity polylines.
We will also give deterministic multiplicative-additive approximation algorithms that,
for any value $\varepsilon>0$, guarantee a solution of size $(1+\varepsilon)\text{OPT}+1$
while running in
$\lO\left(n/\varepsilon\right)$ time
for disks,
axis-aligned line segments and constant-complexity rectilinear polylines,
and $\lO\left(n^{4/3}/\varepsilon\right)$ time for line segments and constant-complexity polylines.
\end{abstract}

\section{Introduction}
\label{sec:Introduction}

We say that two points $s$ and $t$ are \defn{separated}
by a set of objects $C$ if every path from $s$ to $t$
passes through some element of $C$.
The \defn{point-separation problem} asks,
for a given set of (weighted) objects $\cC$ and two points $s$ and $t$,
what is the minimum (weight) subset $C$ of $\cC$ separating $s$ and $t$?
An example of this problem can be found in \cref{fig:example-curves-1}. %

\begin{figure}[h]
\centering
\includegraphics[scale=0.60,page=1]{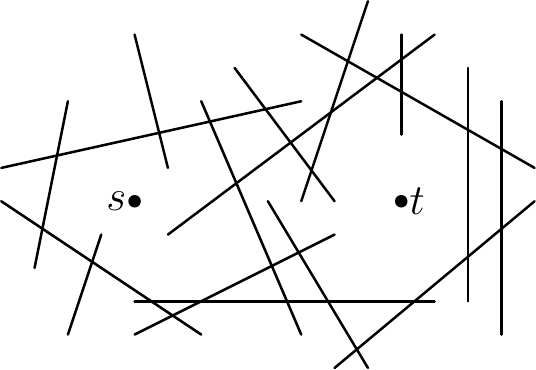}
\hspace{3em}
\includegraphics[scale=0.60,page=3]{simple-example}
\caption{An instance of the point-separation problem with a set of line segment objects (left), and a minimum subset of segments that separate $s$ and $t$ (right). The region of the plane reachable from $s$ without crossing the blocking objects (which does not include $t$) is shaded.}
\label{fig:example-curves-1}
\end{figure}

Point-separation captures a number of real-world scenarios,
particularly problems of detecting or preventing entrances/exits from a region.
For example,
in the case of a small isolated facility in the wilderness,
one may want to install proximity sensors of some kind
(each able to detect movement within a fixed radius)
from a set of candidate installation locations.
Radar sensors can be quite expensive, so it would be helpful to minimize the
number of them required to guarantee detection of anyone approaching.
This can be modelled as an instance of the point-separation problem:
The point $s$ is your facility, the point $t$ is at infinity
(or some point sufficiently far away),
and each candidate sensor location becomes a disk object with the given radius.

An algorithm for point-separation also forms a sub-routine in
a constant-factor approximation-algorithm for the barrier resilience
problem~\cite{KumarLSS21}, motivated by wireless-sensor network design.

Algorithmically, point-separation has been studied by a number of works,
with a few different computational models for representations of the objects~\cite{GibsonKV11,CabelloG16,CabelloM18,KumarLSS21,KumarSoCG2022,spalding2025separating}.
We are primarily interested in the model of \defn{specific object types}:
Each object is of a standard class (e.g., line segments, disks),
and given explicitly.
We will denote the number of planar object $n:|\cC|$ and give running times in terms of this.
Recently, Spalding-Jamieson and Naredla
described a framework obtaining the current fastest algorithms for all previously-studied models~\cite{spalding2025separating}.
Their framework uses a tool called the \emph{homology cover space},
which we will summarize in the next section.
In particular, they define a \defn{geometric intersection graph}
in this space
with two vertices per object,
for $2n$ vertices total.
That is,
a graph whose vertices correspond to connected objects
in the space, and whose edges correspond to pairwise intersections.
They then show that the optimum solution to the point-separation problem
corresponds exactly to finding the minimum of exactly $n$ specific
shortest-path distances in this graph (specifically, the distances between pairs corresponding to the same object),
and they then solve the all-pairs shortest-paths problem
to get these distances.
This direct use of the all-pairs shortest-path problem
inherently requires at least quadratic time
due to the size of its output.

In this work, we are primarily interested in the case of unweighted objects of specific types,
such as disks and line segments,
for which Spalding-Jamieson and Naredla gave a number of new results
by efficiently solving the resulting all-pairs shortest-path problems.
Many of these solutions run in overall times quite close to this quadratic time limit.

Spalding-Jamieson and Naredla also gave fine-grained (conditional) lower
bounds on point-separation for many classes of objects,
including
$\Omega\left(n^{2-\varepsilon}\right)$
lower bounds for many classes of \emph{weighted} objects,
and
$\Omega\left(n^{3/2-\varepsilon}\right)$ lower bounds for many classes
of \emph{unweighted} objects. Notably, both bounds apply for line segments and rectilinear polylines made up of at most $3$ sub-segments.

\subsection{Our Results}
\label{sec:Our Results}

We present numerous results
that we outline here.
First, we will give stronger fine-grained lower bounds for unweighted point-separation
and new fine-grained lower bounds for multiplicative approximations of unweighted point-separation.
Then, our main results are upper bounds that bypass these lower bounds. These are detailed in Table~\ref{tab:results}.
That is, we provide additive $+1$ approximations
and multiplicative-additive $+1$ approximations.
When combined with
algorithms for finding small solutions,
our additive $+1$ approximations
in turn also imply new multiplicative approximation schemes,
at least one of which is (essentially) tight.

Our results apply to a large number of object classes,
including several variations of polylines and polygons.
We say that the \defn{complexity} of a polyline or polygon
is the number of segments that make it up
(for example, a triangle has complexity $3$).
We are primarily concerned with polylines and polygons
of \emph{constant}-complexity,
since they exhibit helpful properties for our purposes.
For instance, a pair of such objects can be tested for intersection in constant time.

\paragraph*{Lower bounds}
While the bounds obtained by Spalding-Jamieson and Naredla for
unweighted objects are weaker than the ones they obtain for weighted objects,
we posit that this is primarily due to their use of
a weak fine-grained hypothesis.
In \cref{sec:mult-approx},
we will show that an improvement of their reduction implies that an exact algorithm
for unweighted point-separation among line segment objects
running in $\OO(n^{2-\varepsilon})$ time for any fixed $\varepsilon>0$
would imply the biggest breakthrough in several decades for the $k$-cycle detection problem.

More relevant to our below work on upper bounds:
We also give lower bounds for multiplicative approximations,
including a essentially tight lower bound for multiplicative approximations
on $\OO(1)$-complexity rectilinear line segments under a
fine-grained hypothesis,
where this lower bound is in terms of the time complexity for solving the $k$-cycle detection problem.
This lower bound matches our corresponding upper bound for multiplicative approximations,
which also relies on an oracle to solve the $k$-cycle detection problem.

\paragraph*{Upper bounds}
The main results of this work
are to show that we can bypass the quadratic-time lower bounds
by considering approximation algorithms
for unweighted point-separation.
Due to the above lower bound,
we primarily construct \emph{additive} approximations
and \emph{multiplicative-additive} approximations.
In particular, the additive approximations we consider are $+1$ additive,
so they are essentially the best-possible subquadratic approximation guarantees possible for point-separation.
Our results also imply subquadratic
(constant-factor) multiplicative approximations as a consequence.

We have a large number of results,
so rather than giving theorem statements here,
our complete set of results is given in \cref{tab:results},
and the corresponding theorem statements can be found in the specified sections.
The main results of this paper are the $\text{OPT}+1$
and $(1+\varepsilon)\text{OPT}+1$ approximation algorithms listed in the table.
These algorithms roughly decompose into four components:
\begin{itemize}
  \item As previously mentioned, Spalding-Jamieson and Naredla showed
    that point-separation can be reduced to the minimum of exactly $n$ shortest-path
    distances in a special graph~\cite{spalding2025separating}.
    We explain their reduction in more detail in \cref{sec:reduction}.
    Then, for our purposes, we show how to efficiently compute individual pairwise shortest-paths
    in this graph for a variety of object classes.
    The results are stated in \cref{sec:reduction},
    but the proofs are left to \cref{sec:sssp}.
    These algorithms are straightforward extensions of
    known techniques for geometric intersection graphs in the plane,
    so we have left the details to the appendix.
    Note that this is not a single result, but rather a collection of results
    for different object classes.
  \item In \cref{sec:paths}, we define the notion of \emph{well-approximated paths},
    and show how to compute them using pairwise shortest-path computations as a black box.
    At a high level, a well-approximated path has the property
    that, if an optimum solution uses some element of the path,
    then we will find a good approximate solution to the point-separation problem
    while constructing the path.
    We have the caveat that this may only apply if the optimum solution is
    somewhat small (depending on what kind of approximation we desire).
  \item If we permit Monte Carlo randomization, then we can find large solutions using
    random sampling, via a simple method described in \cref{sec:monte-carlo}.
  \item Lastly, in \cref{sec:div-conq},
    we show how to use well-approximated paths to perform a divide and conquer
    algorithm.
    At a high level, this algorithm takes a path $\pi$ in the plane from $s$ to $t$,
    and performs divide and conquer on the objects intersecting $\pi$,
    where the division is done using well-approximated paths.
    This divide and conquer method is similar to Reif's algorithm for
    $(s,t)$-min-cut in a planar graph~\cite{reif1983minimum}.
\end{itemize}

The remaining algorithms are derived in
\cref{sec:addk}
and
\cref{sec:mult-approx},
although these primarily leverage machinery created for the above.

\begin{table}
\centering
\renewcommand{\arraystretch}{1.3}
\begin{tabular}{|m{0.27\linewidth}|>{\centering\arraybackslash}m{0.16\linewidth}|>{\centering\arraybackslash}m{0.23\linewidth}|m{0.18\linewidth}|}
\hline
\textbf{Object Type} & \textbf{Running Time} & \textbf{Randomized?} & \textbf{Best-Known Exact}~\cite{spalding2025separating} \\
\hline
\multicolumn{4}{|c|}{\textbf{Approximation Guarantee:} $\mathbf{OPT+1}$ (\cref{sec:div-conq})} \\
\hdashline
Disk & $\OO(n^{3/2} (\log n)^2)$ & Monte Carlo (w.h.p.) & $\OO(n^{2} \log n)$ \\
\hdashline
Line segment & $\lO(n^{11/6})$ & Monte Carlo (w.h.p.) & $\OO(n^{7/3} \log^{1/3} n)$ \\
\hdashline
Rectilinear line segment & $\OO(n^{3/2} (\log n)^2)$ & Monte Carlo (w.h.p.) & $\OO(n^{2} \log \log n)$* \\
\hdashline
$\OO(1)$-complexity polyline & $\lO(n^{11/6})$ & Monte Carlo (w.h.p.) & $\OO(n^{7/3} \log^{1/3} n)$ \\
\hdashline
$\OO(1)$-complexity rectilinear polyline & $\OO(n^{3/2} (\log n)^2)$ & Monte Carlo (w.h.p.) & $\OO(n^{2} \log \log n)$ \\
\hline
\multicolumn{4}{|c|}{\textbf{Approximation Guarantee:} $\mathbf{(1+\varepsilon)\,OPT+1}$ (\cref{sec:div-conq})} \\
\hdashline
Disk & $\OO\left( \frac{n (\log n)^3}{\varepsilon} \right)$ & Deterministic & $\OO(n^{2} \log n)$ \\
\hdashline
Line segment & $\lO\left( \frac{n^{4/3}}{\varepsilon} \right)$ & Deterministic & $\OO(n^{7/3} \log^{1/3} n)$ \\
\hdashline
Rectilinear line segment & $\OO\left( \frac{n (\log n)^3}{\varepsilon} \right)$ & Deterministic & $\OO(n^{2} \log \log n)$ \\
\hdashline
$\OO(1)$-complexity polyline & $\lO\left( \frac{n^{4/3}}{\varepsilon} \right)$ & Deterministic & $\OO(n^{7/3} \log^{1/3} n)$ \\
\hdashline
$\OO(1)$-complexity rectilinear polyline & $\OO\left( \frac{n (\log n)^3}{\varepsilon} \right)$ & Deterministic & $\OO(n^{2} \log \log n)$ \\
\hline
\multicolumn{4}{|c|}{\textbf{Approximation Guarantee:} $\mathbf{OPT+k}$ (\cref{sec:addk})} \\
\hdashline
Disk & $\OO\left( \frac{n^2 \log n}{k} \right)$ & Deterministic & $\OO(n^{2} \log n)$ \\
\hdashline
Line segment & $\lO\left( \frac{n^{7/3}}{k} \right)$ & Deterministic & $\OO(n^{7/3} \log^{1/3} n)$ \\
\hdashline
Rectilinear line segment & $\OO\left( \frac{n^2 \log n}{k} \right)$ & Deterministic & $\OO(n^{2} \log \log n)$ \\
\hdashline
$\OO(1)$-complexity polyline & $\lO\left( \frac{n^{7/3}}{k} \right)$ & Deterministic & $\OO(n^{7/3} \log^{1/3} n)$ \\
\hdashline
$\OO(1)$-complexity rectilinear polyline & $\OO\left( \frac{n^2 \log n}{k} \right)$ & Deterministic & $\OO(n^{2} \log \log n)$ \\
\hline
\multicolumn{4}{|c|}{\textbf{Approximation Guarantee:} $\mathbf{(1+\frac{1}{k})\,OPT}$ (\cref{sec:mult-approx})} \\
\hdashline
Line segment & $\lO(n^{\frac43(2-\frac1{k-1})})$ & Monte Carlo (w.h.p.) & $\OO(n^{7/3} \log^{1/3} n)$ \\
\hdashline
Rectilinear line segment & $\lO(n^{2-\frac1{k-1}})$ & Monte Carlo (w.h.p.) & $\OO(n^{2} \log \log n)$ \\
\hdashline
$\OO(1)$-complexity polyline & $\lO(n^{\frac43(2-\frac1{k-1})})$ & Monte Carlo (w.h.p.) & $\OO(n^{7/3} \log^{1/3} n)$ \\
\hdashline
$\OO(1)$-complexity rectilinear polyline & $\lO(n^{2-\frac1{k-1}})$ & Monte Carlo (w.h.p.) & $\OO(n^{2} \log \log n)$ \\
\hline
\end{tabular}
\caption{Approximation algorithms for the point-separation problem with restricted objects.
Note that we allow polylines
to include or exclude their interior regions (if they have any), and mixing the two types is permitted.
Note that the time complexity for the multiplicative approximations in
\cref{sec:mult-approx}
follow from a known time complexity for $k$-cycle detection,
and would be improved by progress on that problem.
}
\label{tab:results}
\end{table}

\section{Reduction to Shortest-Paths}
\label{sec:reduction}

In this section, we outline the framework of
Spalding-Jamieson and Naredla~\cite{spalding2025separating},
which reduces the point-separation problem to a type of shortest-path problem
in a geometric intersection graph embedded in the space they call the \emph{homology cover}.
We will use some slight simplifications to ease presentation.
At the end, we will also add some new tools for our purposes in this work.

\paragraph*{A simplifying assumption}
Suppose there is a single object $c\in\cC$
that separates $s$ and $t$.
This also includes objects that contain $s$ or $t$.
For all of the object types that we will study
(disks, line segments, and constant-complexity polylines),
it is easy to check in linear time
whether any single such object separates $s$ and $t$.
Henceforth,
we will assume without loss of generality that
the set of candidate objects $\cC$
contains no such object.
Spalding-Jamieson and Naredla
also show that this assumption
can be made in a more general weighted setting~\cite{spalding2025separating}.

\paragraph*{The homology cover}
We now (informally) describe the homology cover space.
Let $\pi$ be the line segment $\overline{st}$.
The set $\RR^2\setminus\{s,t\}$
is the plane with punctures at $s$ and $t$.
The homology cover space
consists of two copies $P_{-1},P_1$
of $\RR^2\setminus\{s,t\}$,
where the path $\pi\setminus\{s,t\}$
in each acts as a portal.
A path entering the portal in $P_i$
exits the portal in $P_{-i}$
at the same point.
For a point $p$ in $\RR^2\setminus\{s,t\}$,
there are exactly two corresponding points $p^{-1},p^1$
in the homology cover space: One in $P_{-1}$,
and one in $P_1$.
Closed planar curves passing through $p$
with an assigned direction
can be put in a natural bijection with curves in the
homology cover space starting at $p^{-1}$
and ending at either $p^1$ or $p^{-1}$.
Importantly, such curves end at $p^1$
if and only if they separate $s$ and $t$,
since this indicates how many times
such planar curves pass through $\pi$.
This is very similar (and related)
to a well-known characterization for testing whether a point lies inside a polygon.
See
\cref{fig:homology-cover-projection}
for examples of two closed planar curves,
one separating $s$ and $t$,
and one not separating $s$ and $t$.
A more formal description of the homology cover space as a manifold
can be found in
\cref{sec:manifold},
or in the work of Spalding-Jamieson and Naredla~\cite{spalding2025separating}.

\begin{figure}[h]
    \centering
    \includegraphics[scale=0.75,page=6]{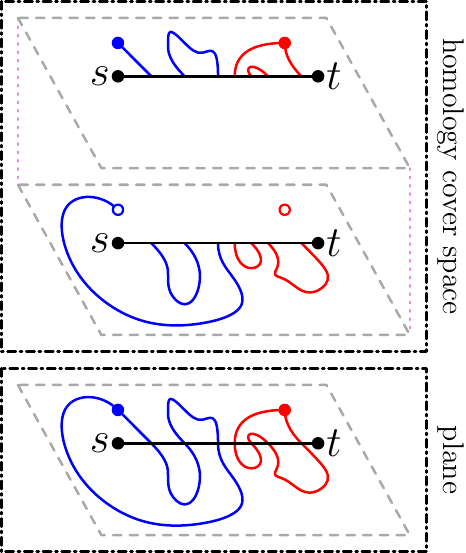}
    \caption{Visualization of two curves in the homology cover,
    and their projections onto the plane.
    One of the resulting projected curves separates $s$ and $t$,
    and the other does not.}
    \label{fig:homology-cover-projection}
\end{figure}

\begin{figure}[t]
    \centering
    \includegraphics[width=0.3\textwidth,page=12]{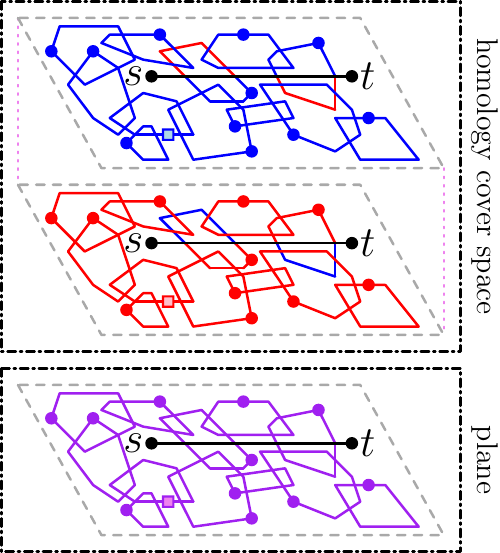}
    \hspace{1em}
    \includegraphics[width=0.3\textwidth,page=17]{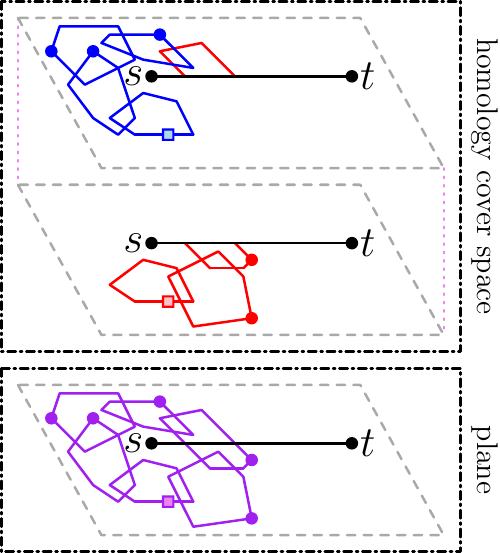}
    \hspace{1em}
    \includegraphics[width=0.3\textwidth,page=22]{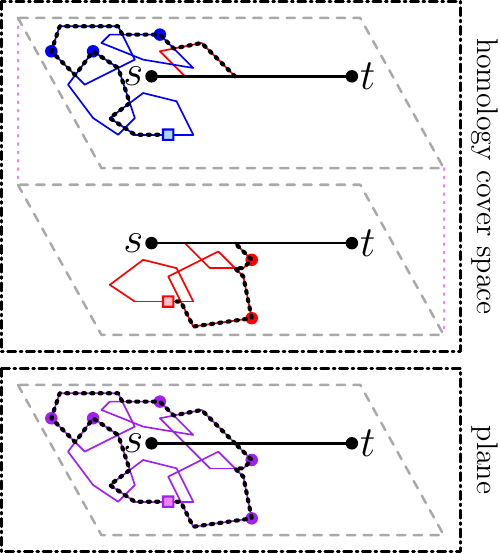}
    \caption{Visualization of objects in the homology cover (left),
    a path $D_c$ through the intersection graph in the homology cover,
    the corresponding set of planar objects $F_c$ (middle, top and bottom, respectively),
    and an induced path whose projection into the plane separates $s$ and $t$ (right).}
    \label{fig:homology-cover-objects-projection}
\end{figure}

\paragraph*{The geometric intersection graph in the homology cover}
Now we describe the geometric intersection graph within the homology cover that we will use.
We will call this graph $\overline G$.
The vertices of a geometric intersection graph correspond to geometric objects,
so in our case we define a set of objects $\overline{\cC}$
in the homology cover based on the given planar objects $\cC$.

For each object $c\in\cC$,
fix a \defn{canonical point} $x_c$.
This point induces exactly two points $x_c^{-1}$ and $x_c^{1}$ in the homology cover,
one in $P_{-1}$ and one in $P_1$.
In some sense, these points will be ``antipodal''.
We assume the objects in $\cC$ are (path-)connected,
so for any point $y$ in $c$,
there is a path $p_y$ in $c$ from $x_c$ to $y$.
Such a path also induces a path starting from $x_c^{i}$
in the homology cover space (for each $i\in\{-1,1\}$).
We let $c^i$ be the object in the homology cover space
defined as the set of endpoints from these paths.
Intuitively, $c^i$ is the object with the same \emph{shape} as $c$
that is induced by the point $x_c^i$.
Note that $|\overline\cC|=2|\cC|$.
See \cref{fig:homology-cover-objects-projection}
for examples of objects in the homology cover space.

Since our objects are path-connected,
any two objects $c^i,d^j\in\overline\cC$
(with $i,j\in\{-1,1\}$)
intersect if and only if
there is a path $p_{c^i,d^j}$
from $x_c^i$ to $x_d^j$
contained entirely in $c^i\cup d^j$.
Let $p_{c,d}$ be the corresponding planar
path from $x_c$ to $x_{d}$
contained entirely in $c\cup d$.
Assume that $p_{c,d}$ is compact\footnote{Not every curve is compact. The topologist's sine curve is the standard example. However, in the case that our objects themselves are defined as closed curves (or have boundaries defined as closed curves), we may always assume that the curve $p_{c,d}$ is compact. This holds for all the object classes we consider.}
and that it never intersects $\overline{st}$
without crossing,
and further assume that it crosses $\overline{st}$
a finite number of times
(see the work of Spalding-Jamieson and Naredla~\cite{spalding2025separating}
for a more careful handling of degenerate cases).
Then, any two objects $c^i$ and $d^j$ intersect
if and only if the following two conditions hold:
\begin{itemize}
    \item $c$ and $d$ intersect in the plane, so there is a planar path $p_{c,d}$ as described.
    \item The number of times $T$ that the path $p_{c,d}$ crosses $\overline{st}$
    has $T\equiv i-j\pmod2$
\end{itemize}

Observe that $\overline{G}$ has an inherent symmetry:
Let $c,d\in\cC$ be objects,
and let $i,j\in\{-1,1\}$.
Then $c^{i}d^{j}\in E(\overline G)$
if and only if
$c^{j}d^{i}\in E(\overline G)$.

\paragraph*{Optimal and approximate solutions to point-separation}
We now highlight some properties of the homology cover and $\overline{G}$ that will be useful for our approximation algorithms.
First, we have the following important property of paths
between vertex pairs
in $\overline{G}$ corresponding to the same planar object:
\begin{observation}[\cite{spalding2025separating}]
\label{obs:path-implies-separation}
For a set of planar objects $\cC$
and planar points $s,t$,
let $\overline{\cC}$ be the corresponding
set of objects in the homology cover space.
Let $c\in\cC$,
and let $c^{-1}$ and $c^1$ be the corresponding objects
in $\overline{\cC}$.
Let $\overline{G}$ be the geometric intersection graph of $\overline{\cC}$ in the homology cover space.
Let $S$ be the set of vertices along some path in $\overline{G}$
from $c^{-1}$ to $c^1$.
Then,
the set of planar objects inducing the vertices in $S$
separates $s$ and $t$.
\end{observation}

See
\cref{fig:homology-cover-objects-projection} (middle and right)
for an example of such a path,
its corresponding planar set of objects,
and a visualization of how the union of such objects
separates $s$ and $t$ in the plane.
While this observation suggests a relationship
between $\overline{G}$ and the point-separation problem,
it does not offer a complete characterization.
To obtain such a characterization, we must consider \emph{shortest}
paths:
\begin{proposition}[\cite{spalding2025separating}]
\label{prop:characterization}
For a set of objects $\cC$,
the graph $\overline G$,
and an object $c\in\cC$,
let $D_c$ denote the elements along an arbitrary shortest-path in $\overline{G}$ from
$c^{-1}$ to $c^1$,
and
let $F_c\subset\cC$ denote the corresponding planar objects
that separate $s$ and $t$.
If $D_c$ is the shortest such path over all $c\in\cC$,
or (equivalently) $F_c$ is the smallest such set over all $c\in\cC$,
then $F_c$ is an optimal solution to the point-separation problem.
Moreover,
for any $c\in\cC$,
and for $d\in F_c$,
we get
$|F_d|\leq|F_c|$ and $|D_d|\leq|D_c|$.
\end{proposition}

The last property guarantees that
any element of the optimum solution itself
induces an optimum solution.

We are
interested in approximation algorithms for point-separation,
so we point out the following important consequence of this result:
\begin{corollary}
\label{cor:basic-approx}
Let $c,d\in \cC$ be objects,
and let $i,j\in\{-1,1\}$.
Suppose that $c^i$ and $d^j$ are at distance $r$ in the intersection graph $\overline G$ in the homology cover.
Then $|D_d|\leq|D_c|+2r$,
and $|F_d|\leq|F_c|+r$.
\end{corollary}

This follows from the symmetry of $\overline{G}$.
That is,
there is a path of length $r$
from $c^i$ to $d^j$,
so there is also a path of length $r$
from $c^{-i}$ to $d^{-j}$,
so the paths from $d^j$ to $c^i$,
from $c^i$ to $c^{-i}$, and from $c^{-i}$ to $d^{-j}$
can be concatenated.
The factor-of-$2$ difference between
the bounds results from the fact that two of these paths
in $\overline{G}$
are induced by the same set of planar objects.
See \cref{fig:homology-cover-objects-projection-approx}
for a simple example.

\begin{figure}[h]
    \centering
    \includegraphics[width=0.4\textwidth,page=24]{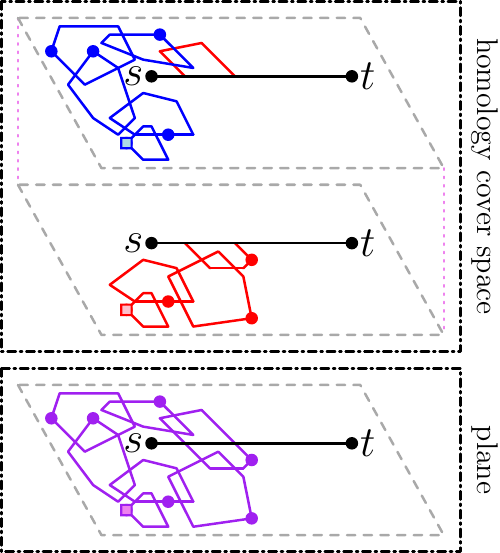}
    \caption{A path through the intersection graph in the homology cover
    that contains the path in \cref{fig:homology-cover-objects-projection} (middle)
    as a sub-path, and the corresponding set of planar objects.
    There are $2$ additional objects in the path in the homology cover,
    but only $1$ additional object in the planar set.}
    \label{fig:homology-cover-objects-projection-approx}
\end{figure}

\subsection{Shortest-Path Algorithms, Notation and Time Complexity}

Many fast shortest-path algorithms
are widely known for a number of geometric intersection graphs in the plane.
For example,
it has recently been shown that
a single-source shortest-path in a disk graph with $n$ vertices
can be found in $\OO(n\log n)$ time~\cite{brewer2025optimal,deberg2025algorithm},
even though the graph itself may have $\Theta(n^2)$ edges.
Spalding-Jamieson and Naredla showed that many algorithms
for \emph{all-pairs} shortest-paths extend to the same object types
in the homology cover~\cite{spalding2025separating}.
However, for this work,
we will require \emph{single-source} shortest-path algorithms.
We will prove the following results in \cref{sec:sssp}:

\begin{theorem}
\label{thm:disk}
For a set of disks $\cC$ in the plane,
and points $s$ and $t$,
let the geometric intersection graph in the homology cover
be denoted $\overline G$.
Then,
there is an algorithm
that computes a single-source shortest-path tree within $\overline G$ in $\OO(n\log n)$ time.
\end{theorem}

\begin{theorem}
\label{thm:sp-linesegs}
For a set $\cC$ of line segments and/or constant-complexity polylines,
and a pair of points $s$ and $t$,
let $\overline G$ be the corresponding intersection graph in the homology cover.
Then,
single-source shortest-paths in $\overline{G}$ can be computed in $\lO(n^{\frac43})$ time,
or $\lO(n)$ time if the segments/polylines are rectilinear.
\end{theorem}

\begin{theorem}
\label{thm:sp-linesegs-rectilinear}
For a set $\cC$ of rectilinear line segments and/or constant-complexity rectilinear polylines,
and a pair of points $s$ and $t$,
let $\overline G$ be the corresponding intersection graph in the homology cover.
Then,
single-source shortest-paths in $\overline{G}$ can be computed in $\OO(n\log n)$ time.
\end{theorem}

All the point-separation algorithms we will describe in this paper
will make use of such single-source shortest-path algorithms as a black-box.
Every such computation will measure the distance from one copy of an object to another,
so we will consistently use the notation defined in
\cref{prop:characterization}: $D_c$ denotes the sequence of elements along the shortest-path in $\overline G$, and $F_c$ denotes the corresponding set of objects separating $s$ and $t$.

Furthermore, in order
to apply SSSP algorithms in a black-box manner,
we will henceforth measure the time complexity of our algorithms in terms of the
number of single-source shortest-path computations performed.
Let $\SP(n')$ denote the time to compute $D_c$ (and, consequently, $F_c$)
for a subgraph of $\overline{G}$ induced by $n'$ objects from $\overline\cC$.
Note that the subsets of objects of $\overline\cC$ we will use will
always contain both or neither element in a pair of corresponding objects.
For instance, if $\cC$ consists only of disk objects, $\SP(n')\in\OO(n'\log n')$.
We will also assume that $\SP(n')$ is super-additive.
That is, for problem sizes $a,b\geq0$, $\SP(a)+\SP(b)\leq\SP(a+b)$.
This super-additivity property will be useful in analyzing a divide and conquer approach later.

\section{Well-Approximated Paths}
\label{sec:paths}

In this section, we will describe and construct
one of our main tools.
Let $\cC$ be a set of planar objects,
let $s,t\in\RR^2$,
and let $\overline{G}$ be the intersection graph in the
homology cover of the corresponding set $\overline{\cC}$.
In order to \emph{exactly} solve the point-separation
problem, it would suffice to compute the minimum value of $|F_c|$
among all $c\in\cC$.
In order to \emph{approximately} solve the point-separation
problem,
it would suffice to find some $c$
so that for every $d\in\cC\setminus\{c\}$,
$(1+\varepsilon)\cdot|F_d|\geq|F_c|$,
where $(1+\varepsilon)\geq1$.
However, we will not directly obtain multiplicative approximations for point-separation.
Rather, we will find \emph{paths} in $\overline{G}$
along which we obtain multiplicative $(1+\varepsilon)$-approximations
for every \emph{sufficiently small} $|F_d|$ value in the path.
The use of paths will play a key role in a divide and conquer method later.
As a consequence,
we can obtain multiplicative-additive $((1+\varepsilon),+1)$-approximations
for all sufficiently small $|F_d|$ values for $d$ in the $1$-hop neighbourhood of such paths
via \cref{cor:basic-approx}.
We have this condition that the values we bound are sufficiently small
so that we can later leverage it to improve some of our final time complexities by a log factor,
by using a different approach to find a large optimal solution.

For an object $c\in\cC$,
a value $(1+\varepsilon)\geq1$,
and a threshold value $k\geq1$,
we say that a \defn{$(1+\varepsilon,k)$-well-approximated path}
is a path in $\overline G$ from $c^{-1}$ to $c^1$
consisting of vertices $c^{-1}=c_0^{b_0},c_1^{b_1},\dots,c_m^{b_m}=c^1$
(where $b_0,\dots,b_m\in\{-1,1\}$),
and a subset of objects $C\subset\{c_0,\dots,c_m\}$
so
that $\min_{i\in\{0,\dots,m\}}(1+\varepsilon)\cdot|F_{c_i}|\geq\min\{(1+\varepsilon)\cdot k,\min_{d\in C}|F_d|\}$.
We call the set $C$ the \defn{net} of the $(1+\varepsilon,k)$-well-approximated path,
since it is analogous to an $\epsilon$-net for a range space.

This structure is useful for the following reason:
If the net $C$ is sufficiently small, then
we efficiently obtain an approximation for \emph{all}
values $F_{c_i}$
by computing the values $F_d$ for $d\in C$.
Note that, in order to get this guarantee,
we only need the smallest value $|F_{d}|$
from the net, not the entire net.
However, the intuition of nets will remain useful for defining and working through the algorithm.

Unfortunately, computing $(1+\varepsilon,k)$-well-approximated paths
with small nets is not straightforward.
We will employ a recursive algorithm that
incrementally and simultaneously computes
the path and the net using
a trial and error approach.
At a high level,
this recursive algorithm will take an object $c$,
and compute $D_c$.
Then, if $D_c$ is itself sufficiently short,
its vertices can be used as the net directly.
Otherwise, it picks a small candidate net
consisting of well-spaced vertices along $D_c$.
With a simple test, we can check if
this potential net is guaranteed to
provide the required approximation guarantees
via the bound given in
\cref{cor:basic-approx}.
If it does, then we have found a net.
However, if it doesn't, then we can extract out some partial progress
and carefully recurse.

\begin{algorithm}[h]
\caption{Computing a $(1+\varepsilon,k)$-well-approximated path}
\label{alg:computepath}
\begin{algorithmic}[1]
\Require Object $c\in\cC$,
useful solution size threshold $k$,
approximation factor $1<(1+\varepsilon)\leq2$
\Ensure Output consists of well-approximated path $P$ with net $C$, and smallest $F_c$ for $c\in C$
\Procedure{WellApproximatedPath}{$c, \overline{G}, k, \varepsilon$}
  \State Compute $D_c$ and $F_c$ in $\overline{G}$ \Comment{One SSSP computation}
  \State $m \gets |F_c|$ \Comment{Note: $1 \leq m \leq n$}
  \If{$m \leq \frac4\varepsilon$} \Comment{\textbf{Case 1}: Small enough to use trivial net}
    \For{each object $c' \in F_c$} \Comment{Uses $\OO\left(\frac1\varepsilon\right)$ SSSP computations}
      \State Compute $F_{c'}$ in $\overline{G}$.
    \EndFor
    \State \Return path $P \gets D_c$, net $C\gets F_c$, and smallest $F_{c'}$ for $c'\in C$
  \Else \Comment{\textbf{Case 2}: Construct a sparse candidate net}
    \State Let $r \gets \left\lfloor\left(1-\frac1{1+\varepsilon}\right)\frac m2\right\rfloor$.
    \State Denote $D_c = [c^{-1}=v_0^{b_0},\, \dots,\, v_m^{b_m}=c^1]$.
    \State Let $s_j$ denote $\min\{j\cdot r,m\}$ for each $0\leq j\leq \lceil\frac mr\rceil$.
    \State Compute $D_{v_{s_j}},F_{v_{s_j}}$ for each $0\leq j\leq\lceil\frac mr\rceil$.
    \Comment{$\OO\left(\frac mr\right)=\OO\left(\frac1\varepsilon\right)$ SSSP computations}
    \If{every $|F_{v_{s_j}}|\geq \min\{\frac m2, k+r\}$} \Comment{\textbf{Case 2a}: Candidate net sufficient}
      \State \Return path $P \gets D_c$, net $C\gets\{v_{s_j}\}_{0\leq j\leq\lceil\frac mr\rceil}$,
      and smallest $F_{c'}$ for $c'\in C\cup\{c\}$
    \Else \Comment{\textbf{Case 2b}: Candidate net may be bad, recurse on first potential violation}
      \State Let $\ell$ be the smallest index where $|F_{v_{s_\ell}}|<\min\left\{\frac m2,k+r\right\}$.
      \State Let $P', C', F_{c'} \gets \Call{WellApproximatedPath}{v_{s_\ell}, \overline{G}, k, \varepsilon}$.
      \State Denote $P' = [u_0,\dots,u_p]$.
      \State Assume labelling $u_0=v_{s_\ell}^{b_{s_\ell}},u_p=v_{s_\ell}^{-b_{s_\ell}}$. \Comment{Requires symmetry}
      \State Let $P \gets [c^{-1}=v_0^{b_0}, v_1^{b_1},\dots, v_{s_\ell}^{b_{s_\ell}}=u_0,\dots,u_p=v_{s_\ell}^{-b_{s_\ell}}, \dots,v_1^{-b_1},v_0^{-b_0}=c^1]$.
      \State \Return concatenated path $P$, net $C\gets C'\cup\{v_{s_j}\}_{0\leq j\leq l}$,
      smallest $F_{c'}$ for $c'\in C$.
    \EndIf
  \EndIf
\EndProcedure
\end{algorithmic}
\end{algorithm}

\begin{figure}
    \centering
    \begin{minipage}[t]{0.46\textwidth}
        \centering
        \includegraphics[scale=0.36,page=1]{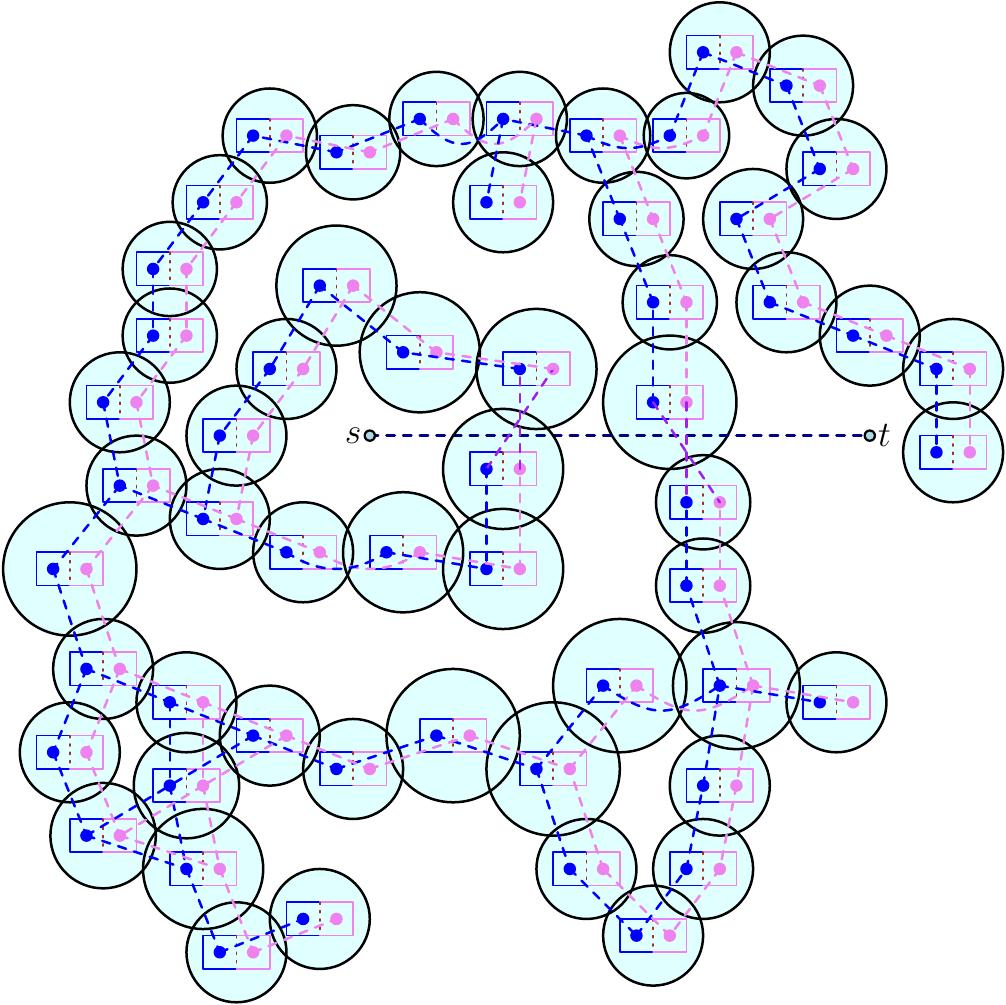}
        \subcaption{The intersection graph $\overline{G}$ in the homology cover,
        using disk centres as canonical points.}
    \end{minipage}
    \hfill
    \begin{minipage}[t]{0.46\textwidth}
        \centering
        \includegraphics[scale=0.36,page=2]{divpaths}
        \subcaption{A shortest-path $D_{c}$.}
    \end{minipage}

    \begin{minipage}[t]{0.46\textwidth}
        \centering
        \includegraphics[scale=0.36,page=3]{divpaths}
        \subcaption{Case 2: A well-spaced candidate net is chosen along $D_{c}$,
        and $D_{c'}$ is computed for each candidate net object $c'$.}
    \end{minipage}
    \hfill
    \begin{minipage}[t]{0.46\textwidth}
        \centering
        \includegraphics[scale=0.36,page=4]{divpaths}
        \subcaption{Case 2b: Recurse on the first $c'$ not providing guarantees.}
    \end{minipage}

    \begin{minipage}[t]{1.0\textwidth}
        \centering
        \includegraphics[scale=0.36,page=6]{divpaths}
        \subcaption{$D_c$ is truncated and then the concatenation of the truncated $D_c$, the recursively found path (in this case $D_{c'}$), and the symmetric path to the truncated component of $D_c$ is returned as the final path.}
    \end{minipage}

    \caption{\Cref{alg:computepath} applied with one recursive call.}
    \label{fig:divpaths}
\end{figure}

The detailed form of the algorithm is given as
\cref{alg:computepath}.
See \cref{fig:divpaths} for an example execution.
We will now prove a series of claims about
\cref{alg:computepath}.

\begin{lemma}
\label{lemma:alg-depth}
\Cref{alg:computepath} terminates at depth $\OO\left(\max\left\{1,\log\left(\min\{m,k+\frac{\varepsilon n}2\}\right)-\log\left(\frac4\varepsilon\right)\right\}\right)$.
\end{lemma}

\begin{proof}
We prove this lemma inductively.
Let $D(m)$ be the depth with $m=|F_c|$.

In cases $1$ and $2a$, the algorithm terminates without any recursive calls,
so $D(m)=1$.

In case $2b$, since $m\leq n$,
it is guaranteed that $|F_{v_{s_\ell}}|<\min\left\{\frac m2,k+\frac{\varepsilon n}2\right\}$.
Hence, $D(m)\leq1+D\left(\min\left\{\frac m2,k+\frac{\varepsilon n}2\right\}\right)$.
The depth given in the lemma statement solves this recurrence.
\end{proof}

\begin{lemma}
\label{lemma:alg-num-calls-per-level}
At each level of recursion, \cref{alg:computepath} makes $\OO\left(\frac1\varepsilon\right)$ SSSP computations.
\end{lemma}

\begin{proof}
We need only derive the fact that $\frac mr\in\OO\left(\frac1\varepsilon\right)$:
$$\frac mr
=\frac{m}{\left\lfloor\left(1-\frac1{1+\varepsilon}\right)\frac m2\right\rfloor}
\leq\frac{2m}{\left(1-\frac1{1+\varepsilon}\right)\frac m2}
=\frac{4}{1-\frac1{1+\varepsilon}}
=\frac{4(\varepsilon+1)}{\varepsilon}\leq\frac8\varepsilon.$$
where we use the assumption that $\varepsilon+1\leq2$ in the last step.
\end{proof}

Note that,
under the assumption that these SSSP computations output the path itself,
the time complexity of each level is also dominated by the SSSP computations.

\begin{proposition}
\label{prop:well-approximated-paths}
\Cref{alg:computepath} uses $\OO\left(\frac1\varepsilon\max\left\{1,\log\min\{m,k+\frac{\varepsilon n}2\}-\log\frac4\varepsilon\right\}\right)$ SSSP computations
to compute a $(1+\varepsilon,k)$-well-approximated path from $c^{-1}$ to $c^1$ in $\overline{G}$.
\end{proposition}

\begin{proof}
The number of SSSP computations follows from
\cref{lemma:alg-depth} and \cref{lemma:alg-num-calls-per-level}.

We will apply
\cref{cor:basic-approx} to prove the approximation guarantee.
Specifically, the net we use will be exactly the set of objects $d$ for
which we computed $F_{d}$ at any point during the algorithm.
It suffices to prove that,
for each object $c_i$ in the final path,
there is some other object $d$ in the final path
for which we computed $F_{d}$,
where
$(1+\varepsilon)\cdot|F_{c_i}|\geq\min\{(1+\varepsilon)\cdot k,|F_d|\}$.
We refer to this as the \emph{approximation guarantee}
for $c_i$.

In case 1, the approximation guarantee is satisfied trivially since we computed
$F_{c'}$ for all the objects $c'$ in the path.

In case 2a, every $|F_{v_{s_j}}|\geq\min\{\frac m2,k+r\}$.
Every $i\in\{0,\dots,m\}$ has $i-\max\{s_j\leq i: 0\leq j\leq \lceil\frac mr\rceil\}\leq r$.
Fix $i$ and fix $s_j$ satisfying this inequality.
Hence, by \cref{cor:basic-approx},
$|F_{v_i}|\geq|F_{v_{s_j}}|-r$.
If $|F_{v_{s_j}}|\geq k+r$, then $|F_{v_i}|\geq k$,
and we are done.
Else, $|F_{v_{s_j}}|\geq\frac m2$,
so
\[|F_{v_i}|
\geq
|F_{v_{s_j}}|-r
\geq
|F_{v_{s_j}}|
-\left(1-\frac1{1+\varepsilon}\right)\frac m2
\geq
|F_{v_{s_j}}|
-\left(1-\frac1{1+\varepsilon}\right)|F_{v_{s_j}}|
=\left(\frac1{1+\varepsilon}\right)|F_{v_{s_j}}|.\]

In case 2b,
we have some smallest value $\ell$
for which $|F_{v_{s_\ell}}|<\min\{\frac m2,k+r\}$.
By the same analysis,
since this is the smallest such value,
we satisfy the approximation guarantees for all $F_{v_i}$ with $i<s_\ell$.
Hence, by induction, the final concatenated path
satisfies all the approximation guarantees.
\end{proof}

\section{A Monte Carlo Algorithm}
\label{sec:monte-carlo}

In this brief section,
we describe a simple randomized Monte Carlo algorithm:

\begin{theorem}
\label{thm:monte-carlo-sampling}
For a given pair of points $s$ and $t$, and a set of objects $\cC$,
there is an iterative algorithm with the following properties:
\begin{itemize}
  \item Each iteration takes $\SP(n)$ time, and produces a separating set of objects.
  \item If there exists a separating set of objects $C\subset\cC$,
    then, with high probability,
    the algorithm will find a separating set of objects $C'\subset\cC$
    with size $|C'|\leq|C|$
    in $\OO\left(\frac{n}{|C|}\log n\right)$ iterations.
\end{itemize}
\end{theorem}

\begin{proof}
The algorithm is quite simple:
In each iteration, a uniformly random object $c$ is sampled from $\cC$,
and then $D_c$ and $F_c$ (defined in \cref{prop:characterization}) are computed.
By \cref{prop:characterization},
a set of objects with size at most $|C|$
will be found if $c\in C$,
so the resulting iteration count is
the number required to sample an element of $C$ with high probability.
\end{proof}

A useful consequence of this theorem is that we can
probabilistically determine if a separating set of size $\OO(\sqrt{n})$ exists
within $\OO(\sqrt{n}\log n)$ iterations.
We will use this consequence in the next section.

\section{Divide and Conquer}
\label{sec:div-conq}

In this section, we will leverage well-approximated paths in order to devise a divide-and-conquer algorithm.
The core divide-and-conquer structure is given in \cref{prop:div-conq}. We then combine this with the single-source shortest path algorithms in \cref{sec:sssp} to obtain the main results of the paper.

\begin{proposition}
\label{prop:div-conq}
Let $s$ and $t$ be points in the plane, and let $\cC$ be a set of
planar objects.
Let $\pi$ be some path from $s$ to $t$.
Assume that each object in $\cC$
intersecting $\pi$ does so a constant number of times,
and that these intersection points can be computed in constant time.
Let $\overline{G}$ denote the intersection graph in the homology cover,
and
let $\SP(n')$ be the time required to compute a shortest-path
through an induced subgraph of $\overline{G}$ with $n'$ vertices.
Assume $\SP(n')$ is super-additive.
Let $\text{OPT}$ denote the size of the optimum solution to the point-separation problem.
Then:
\begin{itemize}
  \item There is a randomized Monte-Carlo algorithm that
    (w.h.p.) produces
    a solution to the point-separation problem
    of size at most $\text{OPT}+1$
    in $\OO\left(\left(\sqrt{n}\log n\right)\cdot\SP(n)\right)$ time.
  \item For any $0<\varepsilon\leq1$,
    there is a deterministic algorithm that produces
    a solution to the point-separation problem
    of size at most $(1+\varepsilon)\text{OPT}+1$
    in $\OO\left(\left(\frac{\log^2 n}{\varepsilon}\right)\SP(n)\right)$~time.
\end{itemize}
\end{proposition}

\begin{proof}
These two results form one proposition since they share a straightforward algorithm,
with the difference that only one of them applies
\cref{thm:monte-carlo-sampling}.

\begin{figure}
    \centering
    \includegraphics[scale=1.0,page=6]{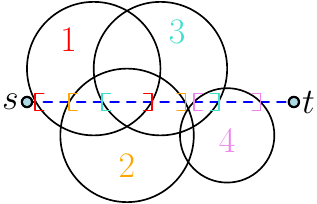}
    \caption{The intersection points of objects with $\overline{st}$ are ordered.}
    \label{fig:intersection-ordering}
\end{figure}

Consider all intersection points, $A$, of all objects intersecting $\pi$.
Note that any feasible solution to the point-separation problem must include at least one of the corresponding objects,
otherwise the path $\pi$ connects $s$ and $t$.
Sort all these objects along $\pi$.
At each stage of the divide and conquer, we will do the following:
\begin{itemize}
  \item Pick the object $c\in\cC$ inducing the (lower) median intersection point of the remaining intersection points in this ordering of $A$.
  \item Compute a $(1+\varepsilon,k)$-well-approximated path $P$ from $c^{-1}$ to $c^1$ in $\overline G$
    using \cref{alg:computepath}, with fixed choices of $\varepsilon$ and $k$.
  \item Remove from $\cC$ all objects in the $1$-hop neighbourhood of $P$, and also remove the corresponding vertices in $\overline G$.
  \item Partition $\cC$ by connected components in $\overline G$ and $A$.
  At most one connected component contains
  a contiguous prefix of the original ordering, and at most one contains a contiguous suffix of the original ordering.
  Let $x_L$ be the first intersection point induced by an object in $P$,
  and let $x_R$ be the last such intersection point.
  One can create a new path $\pi'$ from $s$ to $t$ by concatenating
  the segment $\overline{sx_L}$, a path from $x_L$ to $x_R$ along objects from $P$,
  and the segment $\overline{x_Rt}$.
  Any closed curve in the plane separating $s$ and $t$ must use an object intersecting this path,
  and we have already deleted all the objects intersecting the middle section of this path
  (see \cref{fig:dc-new-path}).
  Let $X_L$ be the sub-ordering of intersection points before $x_L$ corresponding to objects
  that remain undeleted,
  and let $X_R$ be the corresponding sub-ordering for intersection points lying after $x_R$.
  Recurse on these two components with these two sub-orderings.
\end{itemize}
The final output of our algorithm will be the smallest
shortest-path $F_{c'}$ for $c'$ in the nets of any of the well-approximated paths,
which we obtain while running $\cref{alg:computepath}$ during the divide and conquer.

The path $\pi$ plays a key role in the separation of these connected components:
If there were $k$ intersection points intersecting $\pi$ to begin with,
then each new connected component in the residual graph
that is recursed on contains at most $\frac k2$ intersection points. Since each object is required to only intersect $\pi$ a constant number of times, the maximum depth of this recursion is $\OO(\log n)$

\begin{figure}
    \centering
    \includegraphics[scale=1.0,page=1]{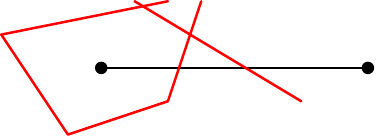}
    \hspace{2em}
    \includegraphics[scale=1.0,page=2]{dc-new-path}
    \caption{Given a planar path from $s$ to $t$, and a separating set of objects, we can generate a new path from $s$ to $t$ that contains a prefix and suffix of the original path, and sub-curves of the objects.
    }
    \label{fig:dc-new-path}
\end{figure}

It therefore follows from
\cref{prop:well-approximated-paths}
that this component of the algorithm takes
$\OO\left(\frac1\varepsilon
\max\left\{1,\log\min\{n,k+\frac{\varepsilon n}2\}-\log\frac4\varepsilon\right\}
\cdot\SP(n)\cdot\log n\right)$
time,
since we can aggregate at each level of the divide and conquer algorithm
and apply super-additivity.

We now show that this divide-and-conquer
algorithm obtains a solution of size at most $\left\lfloor(1+\varepsilon)\text{OPT}+1\right\rfloor$
if $\text{OPT}\leq k$.
At each level of recursion, this approximation guarantee follows from
\cref{prop:well-approximated-paths} for that particular subproblem.
However, In order to get this guarantee on the global problem,
we must translate the approximation guarantees from
our well-approximated paths on subproblems
into approximation guarantees for the global problem,
which requires an additional insight:
For any object $c$
which is in the $1$-hop neighbourhood
of a well-approximated path,
it is computed at some level of recursion
with a subset $C$ of the original objects from $\cC$.
Suppose $c$ is an element of the optimum
solution to the point-separation problem,
so $F_c$ is minimal. If $F_c\subseteq C$ then it is contained entirely in this branch of the recursion and we are done.
If $F_c\not\subseteq C$,
there must be some element
$c'\in F_c\setminus C$
that was discarded at a higher level of recursion.
Assume without loss of generality that
$c'$ was an element discarded at the first level of recursion
among all elements of $F_c\setminus C$.
Let $C'\subset\cC$ be the set of available elements
at the level $c'$ was discarded, so $c'\in C'$.
In this case, $c'\in F_c$.
Since $F_c$ was an optimum solution,
there is some optimum solution $F_{c'}\subset C'$.
Since $c'$ was discarded at this level,
$|F_{c'}|$ must have been well-approximated.

It remains to pick $\varepsilon$ and $k$,
to obtain guarantees,
and to (in one case) apply
\cref{thm:monte-carlo-sampling}.

To obtain the deterministic algorithm, we use $\varepsilon$
as the parameter itself, and use $k=n+1$,
so that the required approximation guarantee follows from the above argument.
In this case, the algorithm
takes $\OO\left(\frac1\varepsilon\SP(n)\log^2n\right)$ time.

Finally, to obtain the Monte-Carlo algorithm,
we choose $\varepsilon=\frac1{1+\sqrt{n}}$,
and $k=\sqrt{n}$.
We also apply $\OO\left(\sqrt{n}\log n\right)$ iterations of
\cref{thm:monte-carlo-sampling},
so any solution of size $\geq k$ is found with high probability
in
$\OO\left(\sqrt{n}\cdot\SP(n)\cdot\log n\right)$ time.
Note that
\begin{align*}
\centering
\log\left(k+\frac{\varepsilon n}2\right)-\log\frac4\varepsilon
&=\log\left(\sqrt{n}+\frac{n}{2(1+\sqrt n)}\right)-\log\left(4\left(1+\sqrt{n}\right)\right)\\
&\leq\log\left(2\sqrt{n}\right)-\log\left(4\sqrt{n}\right)<0,
\end{align*}
so
the divide and conquer step
takes
$\OO\left(\frac1\varepsilon\cdot\SP(n)\cdot\log n\right)=\OO\left(\sqrt{n}\cdot\SP(n)\cdot\log n\right)$
time.
If the optimum solution has size $\text{OPT}\geq k=\sqrt{n}$,
then it was found by the Monte-Carlo step with high probability.
If instead the optimum solution has size $\text{OPT}<k=\sqrt{n}$,
then
the divide and conquer step produces a solution 
(an element of some net) with size at most $\left(1+\frac1{1+\sqrt{n}}\right)\text{OPT}+1$.
Since the solutions are integer-sized, this bound can be reduced to
$\left\lfloor\left(1+\frac1{1+\sqrt{n}}\right)\text{OPT}\right\rfloor+1
=\text{OPT}+1,$
completing both cases of the proof.
\end{proof}

We remark that it was not strictly necessary for the proof that each object intersected $\pi$ a constant number of times.
Rather, it would have sufficed for the total number of intersection points to be $\OO(n)$. However, our applications in the next section will want this additional assumption.

\subsection{Main Theorems}

By applying our algorithms for
single-source shortest-paths, we obtain our main results:

\begin{theorem}
\label{thm:disk-main}
Let $s$ and $t$ be points in the plane, and let $\cC$ be a set of
$n$ disks.

If the optimum solution to the point-separation problem
uses $\text{OPT}$ disks,
then we can find a solution using at most
$\text{OPT}+1$ disks with high probability
in $\OO(n^{\frac32}\log^2n)$ time.

If the optimum solution to the point-separation problem
uses $\text{OPT}$ disks,
then for any $0<\varepsilon\leq1$,
we can find a solution using at most
$(1+\varepsilon)\text{OPT}+1$ disks
in $\OO(\frac1\varepsilon n\log^3n)$ time.
\end{theorem}

\begin{proof}
Apply \cref{prop:div-conq} and \cref{thm:disk}.
\end{proof}

\begin{theorem}
\label{thm:line-seg-convex-poly-main}
Let $s$ and $t$ be points in the plane, and let $\cC$ be a set of
$n$ constant-complexity polylines (i.e.~polylines made up of $\OO(1)$ segments).
We allow any of the polylines enclosing a region to either include or exclude this region.
Assume the optimum solution to the point-separation problem
uses $\text{OPT}$ polylines.

First, we can find a solution to the point-separation problem using at most
$\text{OPT}+1$ polylines (w.h.p.)
in $\lO(n^{\frac{11}6})$ time.

Second, for any $0<\varepsilon\leq1$,
we can find a solution using at most
$(1+\varepsilon)\text{OPT}+1$ polylines
in $\lO(\frac1\varepsilon n^{\frac43})$ time.
\end{theorem}

\begin{proof}
We start by reducing to the case where each of the polylines intersecting the segment $\overline{st}$
includes its interior regions (if it has any).
We refer to the polylines that do not include their interiors as \defn{hollow},
and those that do as \defn{filled}.
Recall that we have assumed no single object in $\cC$ separates $s$ and $t$
(such an object would form an optimal solution, and can be checked for in linear time).
Therefore, no polyline includes exactly one of $s$ and $t$ in an interior region,
whether or not it is filled.

Since we have assumed no single object in $\cC$ separates $s$ and $t$,
there cannot be any polyline with an interior region containing exactly one of $s$ or $t$.
Let $\cC'$ be formed from $\cC$ by modifying the elements of $Q$
to all be hollow.

There is a correspondence between the objects in $\cC$ and $\cC'$.
We claim that subsets separating $s$ and $t$
are also in direct correspondence.
Note that each object in $\cC'$ is a subset of the corresponding object in $\cC$,
so we need only show how to convert a minimum subset of $\cC$ separating $s$ and $t$
into a minimum subset of $\cC'$ separating $s$ and $t$.

Let $C\subset\cC$
be a minimal subset of objects separating $s$ and $t$.
Let $\overline G$ be the geometric intersection
graph in the homology cover induced by $\cC$,
and let
$\overline G'$ be the geometric intersection
graph in the homology cover induced by $\cC'$.
By \cref{prop:characterization},
$C$ is induced by some path $D_c$ (for $c\in C$) in $G$
denoted $c^{-1}=c_0^{b_0},c_1^{b_1},\dots,c_m^{b_m}=c^1$.
Moreover, each adjacent pair of objects $c_i,c_{i+1}$
must intersect in the plane as a pre-requisite.
Suppose one such pair of objects $c_i,c_{i+1}$
intersect, but their corresponding objects in $\cC'$
do not intersect.
Then $c_i$ must lie completely in the interior of $c_{i+1}$, or vice versa.
Therefore, the larger object of the two must intersect a superset of the
objects that the smaller one intersects,
so the smaller object can be discarded from the set $C$
to obtain a smaller separating set of objects.
This contradicts the optimality of $C$, completing the reduction.

Therefore, we can use
\cref{thm:sp-linesegs} (which computes hollow shortest-paths only)
to compute the shortest-paths during the execution of the algorithm
used for \cref{prop:div-conq}.
\end{proof}

\begin{theorem}
\label{thm:axis-aligned-rectangle-main}
Let $s$ and $t$ be points in the plane, and let $\cC$ be a set of
$n$ constant-complexity rectilinear polylines.
We allow any of the polylines enclosing a region to either include or exclude this region.
Assume the optimum solution to the point-separation problem
uses $\text{OPT}$ polylines.

First, we can find a solution to the point-separation problem using at most
$\text{OPT}+1$ rectilinear polylines (w.h.p.)
in $\OO(n^{\frac32}\log^2n)$ time.

Second, for any $0<\varepsilon\leq1$,
we can find a solution using at most
$(1+\varepsilon)\text{OPT}+1$ rectilinear polylines
in $\OO(\frac1\varepsilon n\log^3n)$ time.
\end{theorem}

\begin{proof}
Rectilinear polylines
are a special case of
polylines,
so we can handle the mixed filled/hollow objects as
in the proof of \cref{thm:line-seg-convex-poly-main}.
To get the stated time complexity, apply \cref{prop:div-conq} and \cref{thm:sp-linesegs-rectilinear}.
\end{proof}

\bibliography{references}

\begin{thebibliography}{10}

\bibitem{ayz95}
Noga Alon, Raphael Yuster, and Uri Zwick.
\newblock Color-coding.
\newblock {\em J. ACM}, 42(4):844–856, July 1995.
\newblock \href {https://doi.org/10.1145/210332.210337} {\path{doi:10.1145/210332.210337}}.

\bibitem{ayz97}
Noga Alon, Raphael Yuster, and Uri Zwick.
\newblock Finding and counting given length cycles.
\newblock {\em Algorithmica}, 17(3):209--223, 1997.
\newblock \href {https://doi.org/10.1007/BF02523189} {\path{doi:10.1007/BF02523189}}.

\bibitem{blelloch2008space}
Guy~E. Blelloch.
\newblock Space-efficient dynamic orthogonal point location, segment intersection, and range reporting.
\newblock In {\em Proceedings of the Nineteenth Annual ACM-SIAM Symposium on Discrete Algorithms}, SODA '08, page 894–903, USA, 2008. Society for Industrial and Applied Mathematics.

\bibitem{brewer2025optimal}
Bruce~W. Brewer and Haitao Wang.
\newblock {An Optimal Algorithm for Shortest Paths in Unweighted Disk Graphs}.
\newblock In Anne Benoit, Haim Kaplan, Sebastian Wild, and Grzegorz Herman, editors, {\em 33rd Annual European Symposium on Algorithms (ESA 2025)}, volume 351 of {\em Leibniz International Proceedings in Informatics (LIPIcs)}, pages 31:1--31:8, Dagstuhl, Germany, 2025. Schloss Dagstuhl -- Leibniz-Zentrum f{\"u}r Informatik.
\newblock URL: \url{https://drops.dagstuhl.de/entities/document/10.4230/LIPIcs.ESA.2025.31}, \href {https://doi.org/10.4230/LIPIcs.ESA.2025.31} {\path{doi:10.4230/LIPIcs.ESA.2025.31}}.

\bibitem{CabelloG16}
Sergio Cabello and Panos Giannopoulos.
\newblock The complexity of separating points in the plane.
\newblock {\em Algorithmica}, 74(2):643--663, 2016.
\newblock \href {https://doi.org/10.1007/s00453-014-9965-6} {\path{doi:10.1007/s00453-014-9965-6}}.

\bibitem{CabelloM18}
Sergio Cabello and Lazar Milinkovi\'{c}.
\newblock Two optimization problems for unit disks.
\newblock {\em Comput. Geom.}, 70-71:1--12, 2018.
\newblock \href {https://doi.org/10.1016/j.comgeo.2017.12.001} {\path{doi:10.1016/j.comgeo.2017.12.001}}.

\bibitem{chan2023finding}
Timothy~M. Chan.
\newblock Finding triangles and other small subgraphs in geometric intersection graphs.
\newblock In {\em Proceedings of the 2023 Annual ACM-SIAM Symposium on Discrete Algorithms (SODA)}, pages 1777--1805. SIAM, 2023.

\bibitem{chan2016all}
Timothy~M. Chan and Dimitrios Skrepetos.
\newblock All-pairs shortest paths in unit-disk graphs in slightly subquadratic time.
\newblock In {\em 27th International Symposium on Algorithms and Computation (ISAAC 2016)}, pages 24--1. Schloss Dagstuhl--Leibniz-Zentrum f{\"u}r Informatik, 2016.

\bibitem{ChanS17}
Timothy~M. Chan and Dimitrios Skrepetos.
\newblock All-pairs shortest paths in geometric intersection graphs.
\newblock In Faith Ellen, Antonina Kolokolova, and J{\"{o}}rg{-}R{\"{u}}diger Sack, editors, {\em Algorithms and Data Structures - 15th International Symposium, {WADS} 2017, St. John's, NL, Canada, July 31 - August 2, 2017, Proceedings}, volume 10389 of {\em Lecture Notes in Computer Science}, pages 253--264. Springer, 2017.
\newblock \href {https://doi.org/10.1007/978-3-319-62127-2\_22} {\path{doi:10.1007/978-3-319-62127-2\_22}}.

\bibitem{ChanS19}
Timothy~M. Chan and Dimitrios Skrepetos.
\newblock All-pairs shortest paths in geometric intersection graphs.
\newblock {\em J. Comput. Geom.}, 10(1):27--41, 2019.
\newblock URL: \url{https://doi.org/10.20382/jocg.v10i1a2}, \href {https://doi.org/10.20382/JOCG.V10I1A2} {\path{doi:10.20382/JOCG.V10I1A2}}.

\bibitem{dalirrooyfardDV19graph}
Mina Dalirrooyfard, Thuy~Duong Vuong, and Virginia~Vassilevska Williams.
\newblock Graph pattern detection: hardness for all induced patterns and faster non-induced cycles.
\newblock In {\em Proceedings of the 51st Annual ACM SIGACT Symposium on Theory of Computing}, STOC 2019, page 1167–1178, New York, NY, USA, 2019. Association for Computing Machinery.
\newblock \href {https://doi.org/10.1145/3313276.3316329} {\path{doi:10.1145/3313276.3316329}}.

\bibitem{deberg2025algorithm}
Mark de~Berg and Sergio Cabello.
\newblock An {$O(n\log n)$} algorithm for single-source shortest paths in disk graphs.
\newblock In Anne Benoit, Haim Kaplan, Sebastian Wild, and Grzegorz Herman, editors, {\em 33rd Annual European Symposium on Algorithms (ESA 2025)}, volume 351 of {\em Leibniz International Proceedings in Informatics (LIPIcs)}, pages 81:1--81:15, Dagstuhl, Germany, 2025. Schloss Dagstuhl -- Leibniz-Zentrum f{\"u}r Informatik.
\newblock URL: \url{https://drops.dagstuhl.de/entities/document/10.4230/LIPIcs.ESA.2025.81}, \href {https://doi.org/10.4230/LIPIcs.ESA.2025.81} {\path{doi:10.4230/LIPIcs.ESA.2025.81}}.

\bibitem{deberg2025arxiv}
Mark de~Berg and Sergio Cabello.
\newblock An {$O(n\log n)$} algorithm for single-source shortest paths in disk graphs, 2025.
\newblock URL: \url{https://arxiv.org/abs/2506.07571}, \href {https://arxiv.org/abs/2506.07571} {\path{arXiv:2506.07571}}.

\bibitem{GibsonKV11}
Matt Gibson, Gaurav Kanade, and Kasturi~R. Varadarajan.
\newblock On isolating points using disks.
\newblock In Camil Demetrescu and Magn{\'{u}}s~M. Halld{\'{o}}rsson, editors, {\em Algorithms - {ESA} 2011 - 19th Annual European Symposium, Saarbr{\"{u}}cken, Germany, September 5-9, 2011. Proceedings}, volume 6942 of {\em Lecture Notes in Computer Science}, pages 61--69. Springer, 2011.
\newblock \href {https://doi.org/10.1007/978-3-642-23719-5\_6} {\path{doi:10.1007/978-3-642-23719-5\_6}}.

\bibitem{giyora2007optimal}
Yoav Giyora and Haim Kaplan.
\newblock Optimal dynamic vertical ray shooting in rectilinear planar subdivisions.
\newblock In {\em Proceedings of the eighteenth annual ACM-SIAM symposium on Discrete algorithms}, pages 19--28, 2007.

\bibitem{giyora2009optimal}
Yoav Giyora and Haim Kaplan.
\newblock Optimal dynamic vertical ray shooting in rectilinear planar subdivisions.
\newblock {\em ACM Transactions on Algorithms (TALG)}, 5(3):1--51, 2009.

\bibitem{klost2023algorithmic}
Katharina Klost.
\newblock An algorithmic framework for the single source shortest path problem with applications to disk graphs.
\newblock {\em Computational Geometry}, 111:101979, 2023.

\bibitem{KumarLSS21}
Neeraj Kumar, Daniel Lokshtanov, Saket Saurabh, and Subhash Suri.
\newblock A constant factor approximation for navigating through connected obstacles in the plane.
\newblock In D{\'{a}}niel Marx, editor, {\em Proceedings of the 2021 {ACM-SIAM} Symposium on Discrete Algorithms, {SODA} 2021, Virtual Conference, January 10 - 13, 2021}, pages 822--839. {SIAM}, 2021.
\newblock \href {https://doi.org/10.1137/1.9781611976465.52} {\path{doi:10.1137/1.9781611976465.52}}.

\bibitem{KumarSoCG2022}
Neeraj Kumar, Daniel Lokshtanov, Saket Saurabh, Subhash Suri, and Jie Xue.
\newblock Point separation and obstacle removal by finding and hitting odd cycles.
\newblock In Xavier Goaoc and Michael Kerber, editors, {\em 38th International Symposium on Computational Geometry, SoCG 2022, June 7-10, 2022, Berlin, Germany}, volume 224 of {\em LIPIcs}, pages 52:1--52:14. Schloss Dagstuhl - Leibniz-Zentrum f{\"{u}}r Informatik, 2022.
\newblock URL: \url{https://doi.org/10.4230/LIPIcs.SoCG.2022.52}, \href {https://doi.org/10.4230/LIPICS.SOCG.2022.52} {\path{doi:10.4230/LIPICS.SOCG.2022.52}}.

\bibitem{lincoln2018tight}
Andrea Lincoln, Virginia~Vassilevska Williams, and Ryan Williams.
\newblock Tight hardness for shortest cycles and paths in sparse graphs.
\newblock In {\em Proceedings of the Twenty-Ninth Annual ACM-SIAM Symposium on Discrete Algorithms}, pages 1236--1252. SIAM, 2018.

\bibitem{reif1983minimum}
John~H Reif.
\newblock Minimum s-t cut of a planar undirected network in {$O(n\log^2(n))$} time.
\newblock {\em SIAM Journal on Computing}, 12(1):71--81, 1983.

\bibitem{spalding2025separating}
Jack Spalding{-}Jamieson and Anurag~Murty Naredla.
\newblock Separating two points with obstacles in the plane: Improved upper and lower bounds.
\newblock In Anne Benoit, Haim Kaplan, Sebastian Wild, and Grzegorz Herman, editors, {\em 33rd Annual European Symposium on Algorithms, {ESA} 2025, September 15-17, 2025, Warsaw, Poland}, volume 351 of {\em LIPIcs}, pages 90:1--90:18. Schloss Dagstuhl - Leibniz-Zentrum f{\"{u}}r Informatik, 2025.
\newblock URL: \url{https://doi.org/10.4230/LIPIcs.ESA.2025.90}, \href {https://doi.org/10.4230/LIPICS.ESA.2025.90} {\path{doi:10.4230/LIPICS.ESA.2025.90}}.

\bibitem{spalding2025arxiv}
Jack Spalding-Jamieson and Anurag~Murty Naredla.
\newblock Separating two points with obstacles in the plane: Improved upper and lower bounds, 2025.
\newblock URL: \url{https://arxiv.org/abs/2504.17289}, \href {https://arxiv.org/abs/2504.17289} {\path{arXiv:2504.17289}}.

\bibitem{yuster1994finding}
Raphael Yuster and Uri Zwick.
\newblock Finding even cycles even faster.
\newblock In {\em International Colloquium on Automata, Languages, and Programming}, pages 532--543. Springer, 1994.

\bibitem{yuster1997finding}
Raphael Yuster and Uri Zwick.
\newblock Finding even cycles even faster.
\newblock {\em SIAM Journal on Discrete Mathematics}, 10(2):209--222, 1997.
\newblock \href {https://doi.org/10.1137/S0895480194274133} {\path{doi:10.1137/S0895480194274133}}.

\end{thebibliography}
\pagebreak
\section{A Formal Description of the Homology Cover Space}
\label{sec:manifold}

Let $s,t$ be points in $\RR^2$.
By using the extended plane,
so that the extended plane without $s$ and $t$ is an annulus,
the homology cover space
is exactly the non-trivial $\ZZ_2$-torsor
over the annulus.
We give a more elementary description below.

\begin{figure}[h]
    \centering
    \includegraphics[width=0.3\textwidth]{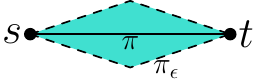}
    \caption{An example of
    $s,t,\pi$
    and $\pi_\epsilon$,
    with $\epsilon=\frac13$.}
    \label{fig:manifold-construction}
\end{figure}

Let $\pi$ be a simple planar path from $s$ to $t$.
Let $\epsilon\in(0,1)$,
and let $\pi_\epsilon$ be the union of open balls
$B\left(x,\epsilon\cdot\min(d(x,s),d(x,t))\right)$ for $x\in\pi$.
Note that the boundary $\partial\pi_\epsilon$
consists of exactly two disjoint paths from $s$ to $t$.
See \cref{fig:manifold-construction} for an example of $s,t,\pi$ and $\pi_\epsilon$.
We can create a manifold consisting of
two copies $P_{-1},P_1$ of
$\RR^2\setminus\pi$,
and two copies $B_{-1},B_1$
of $\pi_\epsilon\cup\partial\pi_\epsilon$.
Each of these subsets of $\RR^2$
has a \defn{top} half, consisting
of all points strictly above the line extension
of $\overline{st}$,
and a \defn{bottom} half,
consisting of all the points strictly below the same line extension.
For $i\in\{-1,1\}$,
we create the manifold by
identifying the top half of $B_{i}$
with the corresponding points in the top half of $P_{i}$,
and
identifying the bottom half of $B_{i}$
with the corresponding points in the bottom half of $P_{-i}$.
The result is exactly what we henceforth call the \defn{homology cover space}.
This manifold is identical regardless of the choice of $\epsilon\in(0,1)$,
so we simplify our notation
by assuming $P_{-1}$ and $P_1$ each include exactly one of the
two copies of $\pi$ as well,
and forgo usage of $B_{-1}$ and $B_1$.

\section{A Simple Deterministic Additive $+k$ Approximation}
\label{sec:addk}

In this short section, we outline a deterministic algorithm
with an additive approximation guarantee surprisingly not attainable using
\cref{prop:div-conq}.

\begin{theorem}
\label{thm:deterministic-big-additive-approx}
For a given pair of points $s$ and $t$, and a set of $n$ objects $\cC$, and a parameter $k$,
there is an algorithm that produces an additive $+k$ approximation for the point-separation problem in $\OO\left(\frac nk\cdot\SP(n)\right)$ time.
\end{theorem}

To prove this, we will use a construction for a well-approximated path
differing from \cref{alg:computepath}.

\begin{proof}
Assume without loss of generality that $\overline G$ is connected
(if not, handle each connected component separately).
Let $T$ be an arbitrary spanning tree of $\overline G$
(for example, a single-source shortest-path tree from an arbitrary vertex).
Consider an Euler tour around $T$,
obtaining a circular sequence covering all vertices.
Cut this circular sequence at some arbitrary point to obtain a sequence $S=(v_1,\dots,v_{4n})$ covering all vertices.
The key property of this sequence is that the distance from $v_i$ to $v_j$ in $\overline G$
is at most $|i-j|$.

In a similar sense to well-approximated paths,
we pick a \emph{net} of vertices in the sequence $S$.
That is, let $S'$ be the subsequence $v_1,v_{1+2k},v_{1+4k},v_{1+6k},\dots,v_{4n}$.
Then, compute $D_c$ and $F_c$ for each $c$ with a vertex in $S'$.
For any $v_i\in S$,
there is some $v_{1+2jk}$ in $S'$
so that $|i-(1+2jk)|\leq k$.
Hence, $|F_{v_i}|\leq|F_{v_{1+2jk}}|+k$ by
\cref{cor:basic-approx},
so the smallest set $F_c$ computed for $c$ with a vertex in $S'$
is an additive $+k$ approximation for the smallest separating set
by \cref{prop:characterization}.

Finally, since we computed a shortest-path only for each element of $S'$ (plus one extra to compute $T$),
and $|S'|\in\OO\left(\frac nk\right)$,
the algorithm runs in the claimed time complexity.
\end{proof}

This algorithm could be sped up for certain classes of objects and small values of $k$,
since there are some techniques in geometric intersection graphs
that can perform shortest-path queries faster if the shortest-path tree
for a nearby vertex (ideally, a neighbouring vertex) is known~\cite{ChanS17}.
However, such speedups would come with a poor tradeoff in the form of an additional factor of $k$ in the time complexity.
Moreover, for sufficiently small values of $k$,
the guarantee of \cref{prop:div-conq} is better.

\section{Fast Shortest-Paths in the Homology Cover}
\label{sec:sssp}

In this section, we discuss methods for computing single-source shortest-paths (SSSP) in the intersection graph in the homology cover, $\overline G$. The object types we focus on in this paper are given in \cref{tab:SSSP}, along with two extra types.
Our results are primarily based on other advances for SSSP in geometric intersection graphs, however small modifications to the algorithms are needed to apply them within the homology cover.

\begin{table}[htbp]
    \centering
    \begin{tabular}{|l|l|l|}
    \hline
    \textbf{Object Type} & $\SP(n)$
    & \textbf{Reference}\\
    \hline
    Disk & $\OO(n \log n)$
    & \cref{thm:disk}\\
    \hline
    Line segment & $\lO(n^{4/3})$
    & \cite{spalding2025separating} \\
    \hline
    $\OO(1)$-complexity polyline & $\lO(n^{4/3})$
    &
    \cite{spalding2025separating} \\
    \hline
    Rectilinear line segment & $\OO(n \log n)$
    & \cref{thm:sp-linesegs-rectilinear}\\
    \hline
    $\OO(1)$-complexity rectilinear polyline & $\OO(n \log n)$
    & \cref{thm:sp-linesegs-rectilinear}\\
    \hline
    \end{tabular}
    \caption{Summary of SSSP running times for geometric intersection graphs in homology cover. All running times here are deterministic.}
    \label{tab:SSSP}
\end{table}

\paragraph*{Disks}
For the intersection graph of disks in the plane, the fastest single-source shortest-path algorithm is quite recent.
Moreover, two groups independently discovered optimal algorithms (optimal for the algebraic decision tree model):
\begin{itemize}
    \item De Berg and Cabello~\cite{deberg2025algorithm} showed that disk graphs
    in the plane admit a structure called a ``clique-based contraction''
    with some helpful properties,
    and that this structure can be computed efficiently.
    \item Brewer and Wang~\cite{brewer2025optimal} instead made use of additively-weighted Voronoi diagrams.
\end{itemize}
Both of these groups' methods attain $\OO(n\log n)$-time algorithms for the single-source shortest-path problem
in a disk graph, including the ability to compute a shortest-path tree.
We will show that the method of de Berg and Cabello can be extended to the intersection graph in the homology cover in the same time complexity.

A \defn{clique-based contraction} of a graph $G=(V,E)$ is a graph $H=(U,F)$
so that every $u\in U$ is exactly a subset of $V$ forming a clique.
Moreover, it is required that
every $v\in V$ is in exactly one $u\in U$ (i.e., $U$ is a partition of $V$),
and that an edge $uu'\in F$ exists in $H$ if and only if there is some edge $vv'\in E$
with $v\in u$ and $v'\in u'$.

De Berg and Cabello showed that,
if $G$ is an intersection graph of $n$ disks $\cC$, a clique-based contraction
with $|F|=\OO(n\log n)$ ``post-contraction'' edges exists,
and that it can be constructed in the same time complexity.
Furthermore, the clique-based contraction
they construct has the property that every $u\in U$
corresponds specifically to a \emph{stabbed clique} in $G$.
That is, for each $u\in U$,
there is some point $p$ in the plane
contained by \emph{every} $v\in u$.
The resulting lemma is:
\begin{lemma}[{\cite[Proposition 8]{deberg2025arxiv}}]
\label{lemma:clique-contract}
    For a set of $n$ disks $\cC$ in the plane,
    let $G$ be the geometric intersection graph of $\cC$.
    Then,
    a clique-based contraction of $G$
    with $\OO(n\log n)$ edges can be computed in $\OO(n\log n)$ time.
\end{lemma}

Another element of their proof involves \defn{bichromatic intersection testing}:
Given two sets $B$ and $R$ of disks (``blue'' and ``red'' disks),
find, for each $b\in B$,
if it intersects any element of $R$ (the element itself need-not be identified).
Well-known methods exist for performing this step among disks in the plane in $\OO((|B|+|R|)\log|R|)$ time~\cite{chan2016all,ChanS17,klost2023algorithmic}.
De Berg and Cabello also showed that, by combining the clique-based contraction of $G$
and the bichromatic intersection testing algorithm,
they can perform shortest-path computations in $G$ in $\OO(n\log n)$ time.
Moreover, these are the only elements needed.
That is, if we can compute a clique-based contraction of the intersection graph $\overline G$ in the homology cover
with the same guarantees,
and perform bichromatic intersection testing on the objects $\overline\cC$ in the homology cover,
then we can solve the single-source shortest-path problem in the same time complexity.
This is exactly the method we will employ:

\begin{theorem}[Restatement of \cref{thm:disk}]
\label{thm:disk-restatement}
For a set of disks $\cC$ in the plane,
and points $s$ and $t$,
let the geometric intersection graph in the homology cover
be denoted $\overline G$.
Then,
there is an algorithm
that computes a single-source shortest-path tree within $\overline G$ in $\OO(n\log n)$ time.
\end{theorem}

\begin{proof}
By the method of de Berg and Cabello~\cite{deberg2025algorithm},
it suffices to construct a clique-based contraction of $\overline G$ with $\OO(n\log n)$ edges
in the same time complexity,
and give an algorithm for static intersection detection in the same time complexity.
In fact,
Spalding-Jamieson and Naredla already gave an algorithm for
bichromatic intersection testing (``static intersection detection'') in the homology cover
with the required time complexity~\cite[Appendix C]{spalding2025arxiv}.

We now show how to compute the clique-based contraction.
First, generate the clique-based contraction of $G$,
in which all the cliques are stabbed cliques.
For a stabbed clique $Q\subset\cC$ with stabbing point $p$,
let $p$ be the canonical point of all $c\in Q$
for the purposes of labelling in the homology cover.
Then, to generate the cliques in $\overline{G}$,
we create two copies of each stabbed clique $Q$,
which we denote $Q^{-1}$ and $Q^{1}$,
relative to the location of the stabbing point $p$.

It remains to compute the edges in the clique-based contraction of $\overline G$
using the clique-based contraction of $G$.
Let $Q$ and $R$ be cliques in the clique-based contraction of $G$,
and let $b_1,b_2\in\{-1,1\}$.
Let the stabbing points of $Q$ and $R$ be denoted $p_Q$ and $p_R$, respectively.
Let $F_Q$, $F_R$, $F_{Q^{b_1}}$ and $F_{R^{b_2}}$ be the union of elements in each clique (de Berg and Cabello call these \defn{flowers}).
Then, $Q^{b_1}$ and $R^{b_2}$ intersect
if and only if there is some point $p$ in the intersection of $F_{Q^{b_1}}$ and $F_{R^{b_2}}$
so that the number of intersections of the length-$2$ polyline
$\overline{p_Rpp_Q}$
with the segment $\overline{st}$
is odd (if $b_1\neq b_2$)
or even (if $b_1=b_2$).
Moreover,
we know that no disk contains $s$ or $t$.
$F_Q$ and $F_R$ intersect if and only if one of two things occurs:
Either the boundaries of $F_Q$ and $F_R$ intersect (in which case we need only consider crossing points $p$),
or one is contained in the other.
We handle these cases separately
to determine if $F_{Q^{b_1}}$ and $F_{R^{b_2}}$ intersect.
For the former case:
de Berg and Cabello~\cite{deberg2025algorithm}
actually iterate through all
crossing points between the boundaries of $F_Q$ and $F_R$
as part of their algorithm,
so we simply check each one.
In the latter case,
we simply look at whether or not the segment $\overline{p_Rp_Q}$
crosses the segment $\overline{st}$,
since
no other path can have a different parity of crossings in this case.
\end{proof}

\paragraph*{Line segments and constant complexity polylines}
Intersection graphs of line segments (and, consequently, constant complexity polylines) in the plane
are known to have small ``biclique covers'':
A \defn{biclique cover} of a graph $G=(V,E)$
is a collection $\{(A_1,B_1),\dots,(A_k,B_k)\}$
where each $(A_i,B_i)$ is a
\defn{biclique} -- $A_i,B_i\subset V$ and $A_i\times B_i\subset E$ --
such that every edge $e\in E$ is in at least one biclique.
The \defn{size} of a biclique cover is the sum $\sum_{i}\left[|A_i|+|B_i|\right]$.
A geometric intersection graph over $n$ line segments (or, consequently, constant complexity polylines) in the plane
has a biclique cover of size $\lO(n^{\frac43})$,
or $\lO(n)$ if the segments are rectilinear~\cite{chan2023finding}.
Moreover, these constructions are also able to compute the covers in the same time complexity (up to poly-logarithmic factors).
Spalding-Jamieson and Naredla also observe that these constructions can be quite easily adapted to line segments in the homology cover~\cite{spalding2025separating}. This mainly uses the fact that the intersection of a line segment and a half-plane is another line segment (which does not apply for disks, hence the method above).

A biclique cover can be thought of as a kind of exact sparsification of a graph.
In particular,
given a biclique cover for a graph $G$ of size $K$,
single-source shortest-paths computations within $G$
can be reduced to single-source shortest-path computations within a directed graph $G'$ containing $\OO(K)$ vertices, even in the case that the vertices are given weights~\cite{spalding2025separating,spalding2025arxiv} (although we are only interested in unweighted objects).
The result is the following theorem:

\begin{theorem}[Restatement of \cref{thm:sp-linesegs}]
\label{thm:sp-linesegs-restatement}
For a set $\cC$ of line segments and/or constant complexity polylines,
and a pair of points $s$ and $t$,
let $\overline G$ be the corresponding intersection graph in the homology cover.
Then,
single-source shortest-paths in $\overline{G}$ can be computed in $\lO(n^{\frac43})$ time,
or $\lO(n)$ time if the segments/polylines are rectilinear.
\end{theorem}

\paragraph*{Rectilinear line segments and constant complexity rectilinear polylines}
For rectilinear line segments in the plane, an even faster approach to single-source shortest path is known.
In particular, Chan and Skrepetos~\cite{ChanS17,ChanS19} considered
the \defn{decremental intersection detection} problem,
in which a set of $n$ objects is given,
and then a sequence of $O(n)$ queries/updates are given,
where each update deletes an object,
and each query provides an object to test for intersection with any element of the set.
If this problem can be solved in $T(n)$ time for a class of objects,
then Chan and Skrepetos showed that a breadth-first search
(i.e. an unweighted single-source shortest-path computation) can be performed
in a geometric intersection graph for the same class of objects.
As Chan and Skrepetos observe,
for axis-aligned (rectilinear) segments,
this problem can be reduced to decremental ray-shooting queries,
and thus can be solved in $\OO(n\log n)$ time by known results~\cite{giyora2007optimal,giyora2009optimal,blelloch2008space}.
Note that this also extends to rectilinear polylines of constant complexity,
since each polyline can be broken into its components.
Using a straightforward method of Spalding-Jamieson and Naredla,
we can extend this result to the homology cover:

\begin{theorem}[Restatement of \cref{thm:sp-linesegs-rectilinear}]
\label{thm:sp-linesegs-rectilinear-restatement}
For a set $\cC$ of rectilinear line segments and/or constant complexity rectilinear polylines,
and a pair of points $s$ and $t$,
let $\overline G$ be the corresponding intersection graph in the homology cover.
Then,
single-source shortest-paths in $\overline{G}$ can be computed in $\OO(n\log n)$ time.
\end{theorem}

\begin{proof}
It suffices to solve decremental intersection detection
among line segments in the homology cover.
As Spalding-Jamieson and Naredla observe~\cite[Appendix C]{spalding2025arxiv},
intersection tests involving $n$ line segments in the homology cover
can be reduced into pairs of intersection tests involving at most $2n$ line segments among two copies of the plane,
so we can apply the method of Chan and Skrepetos~\cite{ChanS17,ChanS19} in each copy of the plane.
\end{proof}

\section{Multiplicative Approximations and $k$-Cycle Detection}
\label{sec:mult-approx}

In this section, we will accomplish two things:
\begin{itemize}
    \item First, we pose a new fine-grained hypothesis for (directed) $k$-cycle detection,
        and show that subquadratic time algorithms for certain forms of point-separation
        are impossible under this hypothesis.
    \item Second, we show that we can obtain arbitrarily good multiplicative approximations
        for point-separation under a scheme that is tight for certain classes of objects under a weaker hypothesis,
        where the time complexity of this scheme itself depends on the time complexity for $k$-cycle detection.
\end{itemize}

The $k$-cycle detection problem asks,
for a given directed graph $G$,
does $G$ admit a cycle of length $k$?
For our purposes, the value $k$ is always assumed to be a constant
whose resulting time-complexity constant-factors will sometimes be part of a function denoted $f(k)$.
In this case, the problem is known to be
equivalent to asking if a directed graph $G$ admits a
\emph{non-simple} cycle of length $k$
(up to logarithmic factors)~\cite{ayz95}.
We shall henceforth refer only to the non-simple variation of the problem.
We are interested in algorithms with complexity in terms of the number $m$ of
edges in the graph $G$.
In this respect, there are two fastest-known algorithms for (non-simple) $k$-cycle detection
(for $k\geq6$),
where the runtime of one depends on the matrix multiplication exponent $\omega\geq2$
(that is, the smallest value $\omega$ so that two $n\times n$ matrices can
be multiplied in $\lO(n^\omega)$ time);
\begin{itemize}
    \item Alon, Yuster and Zwick proposed an algorithm running in $\OO(m^{2-\frac2k})$ time if $k$ is even,
        and $\OO(m^{2-\frac2{k+1}})$ time if $k$ is odd~\cite{ayz95,ayz97}.
    \item Yuster and Zwick~\cite{yuster1994finding,yuster1997finding} proposed an algorithm
        that Dalirrooyfard, Vuong, and Williams~\cite{dalirrooyfardDV19graph} later showed runs
        $\lO(m^{\frac{k\omega-4/k}{2\omega+k-2-4/k}})$ time if $k$ is even,
        and $\lO(m^{\frac{\omega(k+1)}{2\omega+k-1}})$ time if $k$ is odd.
\end{itemize}
Even if it is were shown that the matrix multiplication exponent was $2$,
the second of these algorithms would
still take $\Theta\left(m^{\frac{2k-\frac4k}{2+k-\frac4k}}\right)\subset\Omega\left(m^{2-\frac4k}\right)$ time.%
Although this is a very fundamental problem,
no faster algorithm has been proposed in almost three decades.

Let $TC(m,k)$ be the time-complexity
of the fastest-possible algorithm that can solve
the $k$-cycle problem on a graph with $m$ edges.
Currently, we do not have any reason to believe that $\lim_{k\to\infty}TC(m,k)\in O(f(k)m^{2-\varepsilon}$
for any function $f$,
and such a result would be a very significant breakthrough.
Based on this knowledge, we pose the following fine-grained hypothesis:

\begin{hypothesis}[$k$-Cycle Hypothesis]
There is no algorithm for the $k$-cycle detection problem
running in $\OO(f(k)\cdot m^{2-\varepsilon})$ time for a fixed
$\varepsilon>0$ and function $f$.
That is, $\lim_{k\to\infty}TC(m,k)>f(k) m^{2-\varepsilon}$ for all functions
$f$ and all constants $\varepsilon>0$.
\end{hypothesis}

It should also be noted that Lincoln, Williams, and Williams gave a fine-grained lower bound
for $k$-cycle detection in terms of max-$3$-SAT~\cite[Appendix E.3]{lincoln2018tight},
although their result is weaker than the above hypothesis.
In particular, they show that the Max-$3$-SAT hypothesis
(that no algorithm for max-$3$-SAT exists running in $O(2^{(1-\varepsilon)n})$ time for $n$ variables) implies that $\lim_{k\to\infty}TC(m,k)>\Omega(m^{\frac32-\varepsilon})$ for all $\varepsilon>0$.
Spalding-Jamieson and Naredla used this to show hardness for certain forms
of unweighted point-separation~\cite[Theorem 26]{spalding2025arxiv},
although the bounds we will derive will be stronger.

\paragraph*{A new perspective on exact point-separation and $k$-cycle}

Spalding-Jamieson and Naredla constructed a conditional lower bound
for unweighted point-separation among either line segments or $3$-complexity rectilinear polylines
using a fine-grained reduction from the (non-simple) $k$-cycle detection problem with constant $k$.
They do this by creating $k$ representatives of each directed edge
as objects arranged in $k$ layers.
Since we assume $k$ to be a constant,
there are $\OO(km)=\OO(m)$ such objects in total.
With a slightly simpler form of their construction,
a solution to the point-separation
problem consisting of at most $k$ line segments corresponds directly to
a directed cycle in the original graph.
We demonstrate this simplified form of their construction for line segments in
\cref{fig:lb-examples}.
Consequently, for these object classes,
unweighted point-separation is harder than
$k$-cycle detection, for any constant $k$.
Spalding-Jamieson and Naredla
conclude that there is a conditional $\Omega(n^{\frac32-\varepsilon})$ (for all $\varepsilon>0$)
lower bound for point-separation with line segments (and other classes)
based on the conditional bound of Lincoln, Williams, and Williams.
However, a better view of their reduction also implies some other fine-grained lower bounds:

\begin{proposition}
\label{prop:lb-seg}
If there is an algorithm for point-separation among $n$ line segments
or $3$-complexity rectilinear polylines
running in $\OO(n^{2-\varepsilon})$ time for $\varepsilon>0$,
then there is some $\varepsilon'>0$
for which $k$-cycle detection
in an $m$-edge directed graph can be solved in $\OO(m^{2-\varepsilon'})$ time
for \emph{any} fixed $k$.
In particular, the $k$-cycle hypothesis would be false.
\end{proposition}

\begin{proof}
This result was essentially proven by Spalding-Jamieson and Naredla~\cite[Theorem 26]{spalding2025arxiv}.
\end{proof}

\begin{proposition}
\label{prop:lb-seg-k}
If there is an algorithm for
detecting a solution of size at most $k$ (for $k$ fixed)
in a point-separation instance with $n$ line segments
or $3$-complexity rectilinear polylines
running in time $T(n)$,
then $k$-cycle detection
can be solved in time $O(T(m))$
for a graph with $m$ edges.
\end{proposition}

\begin{proof}
This result follows from the simplified form given in
\cref{fig:lb-examples}
of the lower bound construction by Spalding-Jamieson and Naredla~\cite{spalding2025arxiv}.
Note that this simplified form is still essentially the same as theirs,
and just removes some extra layers of objects that do not change the form of the solution.
\end{proof}

\begin{figure}[h]
\centering
\includegraphics[valign=c,scale=1.00,page=1]{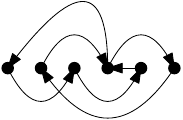}
\hspace{3em}
\includegraphics[valign=c,scale=0.50,page=3]{lower-bound-examples}
\hspace{3em}
\includegraphics[valign=c,scale=0.50,page=5]{lower-bound-examples}
\caption{Left: An instance of $k$-cycle detection in a directed graph (for $k=3$).\\
Middle: A simplified form of the method of Spalding-Jamieson and Naredla for constructing
an instance of point-separation.
$k$ rays are shot out of a point $s$, and a set of $n$ vertex representatives
are chosen along each.
Between consecutive rays, directed edges oriented clockwise are represented as line segments.
The point $t$ is placed somewhere far away.\\
Right: The final instance constructed consists of just the segments representing edges.}
\label{fig:lb-examples}
\end{figure}

\paragraph*{Hardness for $2$-complexity rectilinear polylines}

As explained
in the caption of
\cref{fig:lb-examples},
a simplified form of Spalding-Jamieson and Naredla's reduction from $k$-cycle detection to point-separation among line segments is as follows:
\begin{itemize}
  \item Let $G$ be the input (directed) graph with $n$ vertices and $m$ edges,
    and label its vertices $v_1,\dots,v_n$.
  \item Fix some arbitrary point $s$.
  \item Shoot $k$ unique rays out of $s$.
  \item Pick $n$ points along each ray. Label them $v_1,\dots,v_n$ arbitrarily.
  \item Between consecutive rays, create one segment representing each directed edge, oriented clockwise.
    For instance, for rays $r,r'$ (where the traversal from $r$ to $r'$ is clockwise),
    if there is an edge $v_iv_j\in E(G)$,
    create a segment from the point labelled $v_i$ along $r$ to the point labelled $v_j$ along $r'$.
    These segments form the \emph{only} objects in the point-separation instance.
    That is, the rays and representative points are only used to construct these segments.
  \item Place $t$ sufficiently far away. For example, place it outside the bounding box of all the segments.
\end{itemize}
The important quality for correctness of the reduction
is that if you shoot out $k$ rays from $s$ in between the ones used for the construction,
each one intersects a disjoint subset of objects,
while the original rays separate these subsets and intersect their endpoints.

Since we can let $\pi$ be a path along one of the original rays,
\cref{prop:characterization}
implies we can obtain reductions for other classes of objects
by replacing these line segments with other object types,
so long as the new objects do not cross the original rays,
and only touch the original rays at the same positions as the line segments.

\begin{figure}[h]
\centering
\includegraphics[valign=c,scale=0.50,page=6]{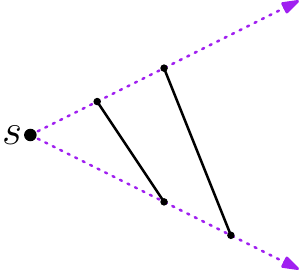}
\hspace{1em}
\includegraphics[valign=c,scale=0.50,page=7]{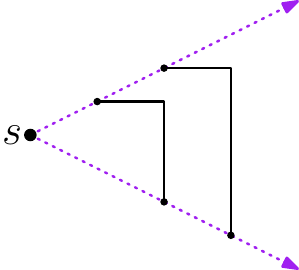}
\hfill
\includegraphics[valign=c,scale=0.50,page=8]{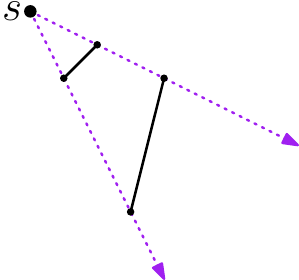}
\hspace{1em}
\includegraphics[valign=c,scale=0.50,page=9]{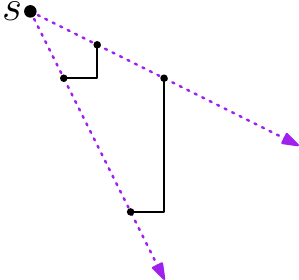}
\caption{Two demonstrations of how line segments between consecutive rays
  can be replaced by $2$-complexity rectilinear polylines.}
\label{fig:rectilinear-polyline-lb}
\end{figure}

Between two rays out of $s$,
for a line segment with an endpoint on each,
we can replace the line segment with a $2$-complexity rectilinear polyline
that crosses neither ray: There are only two $2$-complexity rectilinear polylines
that have the same endpoints as the line segment, and one of them lies entirely in the cone of the rays
(see \cref{fig:rectilinear-polyline-lb}).
Therefore, we obtain the following fine-grained bounds:

\begin{proposition}
\label{prop:lb-2-rect}
If there is an algorithm for point-separation among $n$ $2$-complexity rectilinear polylines
running in
time $T(n)$,
then $k$-cycle detection
in an $m$-edge directed graph can be solved in
$\OO(T(n))$ time
for \emph{any} fixed $k$.
\end{proposition}

\begin{proposition}
\label{prop:lb-2-rect-k}
If there is an algorithm for
detecting a solution of size at most $k$
in a point-separation instance with $n$ $2$-complexity rectilinear polylines
running in time $T(n)$,
then
$k$-cycle detection
in an $m$-edge directed graph can be solved in
$\OO(T(n))$ time for any fixed $k$.
\end{proposition}

\paragraph*{Reductions from point-separation to $k$-cycle}

In the previous subsection,
we discussed reductions from $k$-cycle
to point-separation.
These were not particularly surprising,
since some forms of this were already essentially
shown by Spalding-Jamieson and Naredla~\cite{spalding2025separating}.
However,
we will show in this section that
there is also a partial converse relationship:
Many variants of point-separation
can be reduced \emph{to} $k$-cycle.

By \cref{prop:characterization},
the optimum solution to a point-separation instance with objects $\cC$
is induced by the smallest shortest-path $D_c$
through the geometric intersection graph in the homology cover $\overline G$.

Ignoring the geometric aspects of this problem (for now)
gives us a problem over a graph:
For an undirected graph $G$ with $n$ vertices and $m$ edges, and $r=\OO(n)$ disjoint pairs of vertices $(u_1,v_1),(u_2,v_2),\dots,(u_r,v_r)$,
which pair $(u_i,v_i)$ is the closest in terms of shortest-path distance in $G$?
We call this the \defn{shortest-of-shortest-paths (SoS)} problem.

Rather than study this problem directly directly, we instead define a very slightly more general problem:
For an undirected graph $G$ with $n$ vertices and $m$ edges, a value $d$,
and $r=\OO(n)$ disjoint pairs of vertices $(u_1,v_1),(u_2,v_2),\dots,(u_r,v_r)$,
is any pair of vertices $(u_i,v_i)$ at shortest-path distance at most $d$ in $G$?
We call this the \defn{distance-$d$ detection} problem.

We will show that the distance-$d$ detection problem
is essentially equivalent to the $k$-cycle detection problem
with $k=d$,
in a fine-grained sense.
We prove this with two lemmas, giving the two directions.

\begin{lemma}
If the distance-$d$ detection problem on an undirected graph with $n$ vertices and $m$ edges can be solved
in $\OO(f(n,m,d))$ time,
then the $k$-cycle detection problem (for a fixed $k$)
on a directed graph can be solved in $\OO(f(n,m,k))$ time.
\end{lemma}

To avoid confusion with the similar (but slightly different without another result we state below) notation $TC(m,k)$,
we have used $f(n,m,k)$ and $f(n,m,d)$ here.

\begin{proof}
Let $G$ be a directed graph
that is an instance of the $k$-cycle detection problem.
Construct a new undirected graph $H$ as follows:
For each vertex $v\in V(G)$,
create $k+1$ vertices $v^0,v^1,\dots,v^k$
in $H$.
For each directed edge $uv\in E(G)$,
create $k$ undirected edges $u^0v^1,u^1v^2,\dots,u^{k-1}v^k$
in $H$.

Observe that $H$ has a path from some $v^0$ to $v^k$
(which will inherently have length exactly $k$)
if and only if there is a (directed) cycle of length $k$
in $G$.
Since $H$ also has $(k+1)n=\OO(n)$ vertices
and $km=\OO(m)$ edges,
the result follows.
\end{proof}

\begin{lemma}
\label{lemma:directed-distance-d-detection-to-k-cycle}
If the $k$-cycle detection problem
on a directed graph can be solved in $\OO(f(n,m,k))$ time,
then the distance-$d$ detection problem (for fixed $d$)
on a \emph{directed} graph with $n$ vertices and $m$ edges can be solved
in $\OO(f(n,m,d))$ time.
\end{lemma}

Note that this direction is slightly stronger
in the sense that we give a reduction from \emph{directed}
distance-$d$ detection
(which involves directed distances).

\begin{proof}
The proof is similar to the last one.
Let $G$ be a directed graph,
and let
$(u_1,v_1),\allowbreak(u_2,v_2),\allowbreak\dots,\allowbreak(u_r,v_r)$
(with $r\in\OO(n)$)
be (oriented) pairs of vertices in $G$,
so that these collectively form an instance of the distance-$d$ detection problem.

Create a new graph $H$ as follows:
For each vertex $v\in V(G)$,
create $d$ vertices $v^1,\dots,v^d$.
For each directed edge $uv\in E(G)$,
create directed edges $u^1v^2,u^2v^3,\dots,u^{d-1}v^d$.
If $v$ appears in some pair $(w,v)$,
then also create a directed edge $u^dw^1$.

We claim that there is some $(u_i,v_i)$ at distance
at most $d$ in $G$ if and only if there is a $d$-cycle
in $H$.

In the forward direction,
let the path be denoted $u_i=w_0,w_1,\dots,w_d=v_i$.
Then the corresponding cycle in $H$
is exactly $w_0^1,w_1^2,\dots,w_{d-1}^d$.
Note that there is an edge in $G$ between $w_{d-1}$ and $w_d=v_i$,
so there is an edge in $H$ between $w_{d-1}^d$ and $w_0^1=u_i^1$.

In the backward direction,
any such cycle must go through exactly one vertex on each layer,
so let the vertices be denoted
$w_0^1,w_1^2,\dots,w_{d-1}^d$.
By definition, there must be some $w_d$ in $G$
for which $(w_0,w_d)$ is among the pairs,
and $w_{d-1}w_d\in E(G)$.
The sequence $w_0,w_1,\dots,w_d$
is therefore a length-$d$ shortest-path between $(w_0,w_d)$ in $G$.
\end{proof}

This second lemma serves as the basis for our partial converses to
\cref{prop:lb-seg-k}
and
\cref{prop:lb-2-rect-k}.
It is also critical that we proved it for \emph{directed} distance-$d$ detection.
This is because, similarly to \cref{sec:sssp}, we can use biclique covers
to reduce the detection of an at-most-size-$k$ solution in point-separation
to an instance of \emph{directed} distance-$d$ detection,
for $d=2k$.
Recall that we use the notation $TC(m,k)$
to denote the time complexity of
the $k$-cycle detection problem
on a (directed) graph with $m$ edges.
We obtain the following results:

\begin{theorem}
\label{thm:small-sln-detect}
Let $\cC$ be a set of $n$ objects.
Let $s$ and $t$ be points.
Let $k$ be a constant.

If $\cC$ is a set of line segments
or $\OO(1)$-complexity polylines,
then it can be checked whether or not $s$ and $t$
are separated by at most $k$
of the objects in $\lO(TC(n^{\frac43}\cdot\text{polylog}(n),k))$ time.

If $\cC$ is a set of rectilinear line segments
or $\OO(1)$-complexity rectilinear polylines,
then it can be checked whether or not $s$ and $t$
are separated by at most $k$
of the objects in $\lO(TC(n\cdot\text{polylog}(n),k))$ time.
\end{theorem}

\begin{proof}
Let $\overline{G}$
be the intersection graph in the homology cover space
of the projected objects $\cC$.
By the same arguments as those in \cref{sec:sssp},
we can reduce the problem of checking whether or not $s$ and $t$
can be separated by at most $k$
of the objects
to an instance of distance-$2k$ detection
using a biclique cover to generate a (sparse) directed graph.

If $\cC$ is a set of line segments
or $\OO(1)$-complexity polylines,
then the resulting graph has $\lO(n^{\frac43})$ edges
(and can be found in this same time complexity).
By
\cref{lemma:directed-distance-d-detection-to-k-cycle},
we can then solve the problem in $\lO(TC(n^\frac43\cdot\text{polylog}(n),k))$ time.

If $\cC$ is a set of rectilinear line segments
or $\OO(1)$-complexity rectilinear polylines,
then the resulting graph has $\lO(n)$ edges
(and can be found in this same time complexity).
By
\cref{lemma:directed-distance-d-detection-to-k-cycle},
we can then solve the problem in $\lO(TC(n\cdot\text{polylog}(n),k))$ time.
\end{proof}

By applying the result of Alon, Yuster and Zwick~\cite{ayz95},
we can also obtain a statement with precise time complexities:

\begin{corollary}
\label{cor:small-sln-detect}
Let $\cC$ be a set of $n$ objects.
Let $s$ and $t$ be points.
Let $k$ be a constant.

If $\cC$ is a set of line segments
or $\OO(1)$-complexity polylines,
then it can be checked whether or not $s$ and $t$
are separated by at most $k$
of the objects in $\lO(n^{\frac43(2-\frac1k)})$ time.

If $\cC$ is a set of rectilinear line segments
or $\OO(1)$-complexity rectilinear polylines,
then it can be checked whether or not $s$ and $t$
are separated by at most $k$
of the objects in $\lO(n^{2-\frac1k})$ time.
\end{corollary}

\paragraph*{Upper and lower bounds for multiplicative approximations}

We have shown some careful reductions between point-separation problems
(specifically, checking for solutions of at most a certain fixed size)
and the $k$-cycle detection problem in directed graphs.
When combined with our main additive approximations from earlier in the paper,
these results also imply both upper and lower bounds for multiplicative approximations,
as we will show in this subsection.

We start with the lower bound:
\begin{theorem}
Let $k$ be a constant.
If there were
a multiplicative $(1+\frac1{k+1})$-approximation for
point-separation among $n$ line segments
or $2$-complexity rectilinear polylines
running in $\OO(T(n))$ time,
then
we would be able to solve
the directed $k$-cycle detection problem
in a graph with $m$ edges
in $\OO(T(m))$.
\end{theorem}

\begin{proof}
Note that,
if the optimum solution to an (unweighted) point-separation instance
has size at most $k$,
then a $(1+\frac1{k+1})$-approximation
is exactly the optimal solution, always.
Therefore, the result follows from
\cref{prop:lb-seg-k}
and
\cref{prop:lb-2-rect-k}.
\end{proof}

The upper bound is as follows:
\begin{theorem}
Let $\cC$ be a set of $n$ objects.
Let $s$ and $t$ be points.
Let $k\geq3$.

If $\cC$ is a set of line segments
or $\OO(1)$-complexity polylines,
then we can give a randomized Monte Carlo $(1+\frac1k)$-approximation
(succeeding w.h.p.)
to point-separation running in
$\lO(n^{\frac{11}6}+TC(n^{\frac43}\cdot\text{polylog}(n),k))$
time.

If $\cC$ is a set of rectilinear line segments or
$\OO(1)$-complexity rectilinear polylines,
then we can give a randomized Monte Carlo $(1+\frac1k)$-approximation
(succeeding w.h.p.)
to point-separation running in
$\lO(n^{\frac32}+TC(n\cdot\text{polylog}(n),k))$
time.
\end{theorem}

\begin{proof}
Apply the appropriate result from
\cref{thm:line-seg-convex-poly-main}
and
\cref{thm:axis-aligned-rectangle-main},
to get an additive $+1$ approximation to the problem.
We therefore have a $(\frac{k+1}k)=(1+\frac1k)$-approximation
if the solution has size at least $k$.
To get this multiplicative factor in general,
we must also detect solutions strictly smaller than size $k$.
We use
\cref{thm:small-sln-detect}
to accomplish this.
\end{proof}

In particular, since
the max-$3$-SAT hypothesis
implies that $\lim_{n\to\infty}TC(n,k)\geq\Omega(n^{\frac32-\varepsilon})$
for all $\varepsilon>0$,
we have the following result:

\begin{corollary}
Under the max-$3$-SAT hypothesis,
if $T_{\text{PS}}(n,1+\frac1k)$ is the time
to give a $(1+\frac1k)$-approximation for point-separation among $\OO(1)$-complexity polylines,
then $\lim_{k\to\infty}f_1(k)\cdot TC(n,k)\lesssim
\lim_{k\to\infty}T_{\text{PS}}(n,1+\frac1k)\lesssim
\lim_{k\to\infty}f_2(k)\cdot TC(n\cdot\text{polylog}(n),k)$
for some functions $f_1,f_2$,
where $\lesssim$ here hides sub-polynomial factors.
\end{corollary}

\appendix

\end{document}